\documentclass[smallcondensed]{svjour3}     
\smartqed  

\usepackage{graphicx}
\usepackage{epsfig}
\usepackage{rotating}
\usepackage{amssymb}
\usepackage{amsmath,amsfonts}
\usepackage[colorlinks=false]{hyperref}

\usepackage[utf8]{inputenc}
\usepackage{verbatim}
\usepackage[tableposition=top,font=small,labelfont=bf,format=hang]{caption}
\usepackage{booktabs}
\usepackage{diagbox}
\usepackage{enumerate}

\allowdisplaybreaks[1]

\renewcommand{\vec}[1]{\boldsymbol{#1}} 

\DeclareMathOperator{\lcm}{\text{lcm}}

\begin{document}

\title{$\mathbb{Z}_p\mathbb{Z}_{p^2}$-additive cyclic codes: kernel and rank}


\author{  Xuan Wang        \and
        Minjia Shi
}


\institute{X. Wang  \at
              Anhui University, Hefei, China\\
              \email{wang\_xuan\_ah@163.com}
           \and
           M. Shi \at
              Anhui University, Hefei, China \\
              \email{smjwcl.good@163.com}
}

\date{Received: date / Accepted: date}

\maketitle

\begin{abstract}
	A code $C = \Phi(\mathcal{C})$ is called $\mathbb{Z}_p \mathbb{Z}_{p^2}$-linear if it's the Gray image of the $\mathbb{Z}_p \mathbb{Z}_{p^2}$-additive code $\mathcal{C}$.
	In this paper, the rank and the dimension of the kernel of $\mathcal{C}$ are studied.
	Both of the codes $\langle \Phi(\mathcal{C}) \rangle$ and $\ker(\Phi(\mathcal{C}))$ are proven $\mathbb{Z}_p \mathbb{Z}_{p^2}$-additive cyclic codes,
	and their generator polynomials are determined.
	Finally, accurate values of rank and the dimension of the kernel of some classes of $\mathbb{Z}_p \mathbb{Z}_{p^2}$-additive cyclic codes are considered.
\end{abstract}

\section{Introduction} \label{sec:introduction}

Denote by $\mathbb{Z}_p$ and $\mathbb{Z}_{p^k}$ be the rings of integers modulo $p$ and $p^k$, respectively.
Let $\mathbb{Z}_p^n$ and $\mathbb{Z}_{p^k}^n$ be the space of $n$-tuples over $\mathbb{Z}_p$ and $\mathbb{Z}_{p^k}$, respectively.
A $p$-ary code is a non-empty subset $C$ of $\mathbb{Z}_p^n$.
If the subset is a vector space, then we say $C$ is a linear code.
Similarly, a non-empty subset $\mathcal{C}$ of $\mathbb{Z}_{p^k}^n$ is a linear code if $\mathcal{C}$ is a submodule of $\mathbb{Z}_{p^k}^n$.

In $1994$, Hammons \textit{et al.} \cite{1994} showed that some well-known codes can be seen as Gray images of linear codes over $\mathbb{Z}_4$.
Later, the study of cyclic codes and $\mathbb{Z}_4$-codes develops.
Specially, Calderbank \textit{et al.} \cite{p-adic_cyclic_codes} gave the structure of cyclic codes over $\mathbb{Z}_{p^k}$.
Besides, Ling \textit{et al.}  \cite{LS_Gray_map} studied $\mathbb{Z}_{p^k}$-linear codes, and characterized the linear cyclic codes over $\mathbb{Z}_{p^2}$ whose Gray images are linear cyclic codes.

In $1973$, Delsarte \cite{association_schemes} first defined and studied the additive codes in terms of association schemes.
Borges \textit{et al.} \cite{binary_images_of_z2z4_cyclic,z2z4_generators} studied the generator polynomials, dual codes
and binary images of the $\mathbb{Z}_2 \mathbb{Z}_4$-additive codes.
Since then, a lot of work has been devoted to characterizing $\mathbb{Z}_2\mathbb{Z}_{4}$-additive codes.
Dougherty \textit{et al.} \cite{1-weight-Z2Z4} constructed one
weight $\mathbb{Z}_2\mathbb{Z}_{4}$-additive codes and analyzed their parameters.
Benbelkacem \textit{et al}. \cite{Z2Z4-ACD-LCD} studied
$\mathbb{Z}_2\mathbb{Z}_{4}$-additive complementary dual codes and their Gray images.
In fact, these codes can be viewed as a generalization of the linear complementary dual (LCD for short) codes \cite{LCD} over finite fields.
Bernal \textit{et al.} \cite{Z2Z4-decoding} introduced a decoding method of the $\mathbb{Z}_2\mathbb{Z}_{4}$-linear codes.
More structure properties of $\mathbb{Z}_2\mathbb{Z}_{4}$-additive codes can be found in \cite{Z2Z4-MDS,Z2Z4-kernel-and-rank}.

Moreover, the additive codes over different mixed alphabet have also been intensely studied, for example, $\mathbb{Z}_2\mathbb{Z}_{2}[u]$-additive codes \cite{Z2Z2u}, $\mathbb{Z}_2 \mathbb{Z}_{2^s}$-additive codes \cite{Z2Z2s}, $\mathbb{Z}_{p^r} \mathbb{Z}_{p^s}$-additive codes \cite{ZprZps}, $\mathbb{Z}_2\mathbb{Z}_{2}[u, v]$-additive codes \cite{Z2Z2uv}, $\mathbb{Z}_p \mathbb{Z}_{p^k}$-additive codes \cite{zpzpk_additive_and_dual} and $\mathbb{Z}_p (\mathbb{Z}_p + u \mathbb{Z}_p)$-additive codes \cite{WRS}, and
so on.
It is worth mentioning that $\mathbb{Z}_2\mathbb{Z}_{4}$-additive cyclic codes form an important family of $\mathbb{Z}_2\mathbb{Z}_{4}$-additive codes, many optimal binary codes can be obtained from the images of this family of codes.
In $2014$, Abualrub \textit{et al.} \cite{z2z4_additive_cyclic_codes} discussed $\mathbb{Z}_2\mathbb{Z}_{4}$-additive cyclic codes.
In $2021$, Shi \textit{et al.} \cite{z2z4-quasi-cyclic} defined the $\mathbb{Z}_2 \mathbb{Z}_4$-additive quasi-cyclic codes and gave the conditions for the code to be self-dual and additive complementary dual (ACD for short), respectively.
More details of $\mathbb{Z}_2\mathbb{Z}_{4}$-additive cyclic codes can be found in \cite{z2z4_additive_cyclic_codes,z2z4_generators,binary_images_of_z2z4_cyclic}.

In $2010$, Fern\'{a}dez \textit{et al.} \cite{z2z4-linear} studied the rank and kernel of $\mathbb{Z}_2 \mathbb{Z}_4$-linear codes,
where $\mathbb{Z}_2 \mathbb{Z}_4$-linear codes are the images of $\mathbb{Z}_2 \mathbb{Z}_4$-additive codes under the generalized Gray map $\Phi$.
Moreover, the rank and the kernel have been intensely studied, for example cyclic and negacyclic quaternary codes \cite{kernel_of_4-ary_code},
generalized Hadamard codes \cite{rank_and_kernel_of_GH_matrices}, additive generalised Hadamard codes \cite{rank_and_kernel_of_additive_GH_code} and so on.
Note that, in $2019$, Borges \textit{et al.} \cite{z2z4_cyclic_kernel} studied the rank and the kernel of $\mathbb{Z}_2 \mathbb{Z}_4$-additive cyclic codes. In addition, in \cite{z2z4-linear}, the authors discussed the rank and kernel dimension of the $\mathbb{Z}_2 \mathbb{Z}_4$-linear codes.
Later, Shi \textit{et al.} \cite{wsk_z3z9_linear_rank_kernel} generated the results about rank and kernel to $\mathbb{Z}_3 \mathbb{Z}_9$.

It is a natural problem: what the results will be for $\mathbb{Z}_p \mathbb{Z}_{p^2}$-additive cyclic codes with an odd prime $p$? Motivated by \cite{z2z4_cyclic_kernel}, \cite{wsk_z3z9_linear_rank_kernel} and \cite{Xuan-Gray-images}, this paper is devoted to studying the kernel and the rank of $\mathbb{Z}_p\mathbb{Z}_{p^2}$-additive cyclic codes.
The paper is organized as follows.
In Section \ref{sec:preliminaries}, we recall the necessary concepts and properties on $\mathbb{Z}_p\mathbb{Z}_{p^2}$-additive cyclic codes and Gray maps.
In Section \ref{sec:kernel} and Section \ref{sec:rank}, we study the kernel $\mathcal{K}(\mathcal{C})$ and the span $\mathcal{R}(\mathcal{C})$
of a $\mathbb{Z}_p\mathbb{Z}_{p^2}$-additive cyclic code $\mathcal{C}$, respectively.
In Section \ref{sec:pairs and classification}, the values of the rank and the dimension of the kernel of some codes are given.
In Section \ref{sec:comparison}, we discuss the difference between the results and general $\mathbb{Z}_p \mathbb{Z}_{p^2}$-linear codes.

\section{Preliminaries} \label{sec:preliminaries}
In the following sections, assume $p$ an odd prime number and $\gcd(\beta, p) = 1$.
\subsection{$\mathbb{Z}_p\mathbb{Z}_{p^2}$-Additive Cyclic Codes}\label{subsec:ZpZp2 additive cyclic codes}
In this subsection, let $\mathcal{C}$ be a $\mathbb{Z}_p\mathbb{Z}_{p^2}$-additive code of length $\alpha + \beta$,
which is a subgroup of $\mathbb{Z}_p^{\alpha} \times \mathbb{Z}_{p^2}^{\beta}$.
Then $\mathcal{C}$ is called a $\mathbb{Z}_p\mathbb{Z}_{p^2}$-\textbf{additive cyclic} code if for any codeword $\vec{u} = (\vec{u}', \vec{u}'') = (u'_0, u'_1, \cdots, u'_{\alpha - 1}, u''_0, u''_1, \cdots, u''_{\beta - 1}) \in \mathcal{C}$, its cyclic shift
\[
\pi(\vec{u}) = (\pi(\vec{u}'), \pi(\vec{u}'')) = (u'_{\alpha - 1}, u'_0, \cdots, u'_{\alpha - 2}, u''_{\beta - 1}, u''_0, \cdots, u''_{\beta - 2})
\]
is also in $\mathcal{C}$.

There exists a bijection between $\mathbb{Z}_p^{\alpha} \times \mathbb{Z}_{p^2}^{\beta}$ and
$R_{\alpha, \beta} = R_{\alpha} \times R_{\beta} = \mathbb{Z}_p[x]/(x^{\alpha}-1) \times \mathbb{Z}_{p^2}[x]/(x^{\beta}-1)$
given by
\[
(u'_0, \cdots, u'_{\alpha - 1}, u''_0, \cdots, u''_{\beta - 1}) \longmapsto
(u'_0 + \cdots + u'_{\alpha - 1} x^{\alpha-1}, u''_0 + \cdots + u''_{\beta - 1} x^{\beta-1}).
\]
Hence, any codeword of $\mathcal{C}$ can be regarded as a vector or as a polynomial.
Denote the $p$-ary reduction of a polynomial $h \in \mathbb{Z}_{p^2}[x]$ by $\overline{h} \equiv h  \pmod{p}$.
Then define the following multiplication \[ h  \star (c_1, c_2) = ( \overline{h} c_1, h c_2), \]
where $(c_1, c_2) \in R_{\alpha, \beta}$ and the products of the right side are the standard polynomial products in $\mathbb{Z}_p[x]/(x^{\alpha}-1)$
and $\mathbb{Z}_{p^2}[x]/(x^{\beta}-1)$.

Let $\kappa_1$ be the dimension of the subcode $\{(\vec{u}_1 | \vec{0}) \in \mathcal{C} \}$ and $\kappa_2 = \kappa - \kappa_1$.
Similarly, let $\delta_2$ be the dimension of the subcode $\{ (\vec{0} | \vec{v}_2) \in \mathcal{C}\}$ whose codewords are all of order $p^2$
and $\delta_1 = \delta - \delta_2$.

A $\mathbb{Z}_p\mathbb{Z}_{p^2}$-additive code $\mathcal{C}$ is called \textbf{separable} if $\mathcal{C} = \mathcal{C}_X \times \mathcal{C}_Y$.
By definition, a $\mathbb{Z}_p\mathbb{Z}_{p^2}$-additive code is separable if and only if $\kappa_2 = \delta_1 = 0$; that is $\kappa = \kappa_1$ and $\delta = \delta_2$.

\begin{theorem}\cite{ZprZps}
	\label{generator matrix of zpzp2}
	Let $\mathcal{C}$ be a $\mathbb{Z}_p\mathbb{Z}_{p^2}$-additive code of type $(\alpha, \beta; \gamma, \delta; \kappa)$.
	Then $\mathcal{C}$ is permutation equivalent to a $\mathbb{Z}_p\mathbb{Z}_{p^2}$-additive code with the canonical generator matrix of the form
	\begin{equation}
		\label{eq:generator matrix}
		\mathcal{G}_S =
		\begin{pmatrix}
			\begin{array}{cc|ccc}
				I_\kappa & T' & pT_2 & \vec{0} & \vec{0} \\
				\vec{0} & \vec{0} & pT_1 & pI_{\gamma - \kappa} & \vec{0} \\
				\hline
				\vec{0} & S' & S & R & I_\delta
			\end{array}
		\end{pmatrix},
	\end{equation}
	where $I_n$ is the identity matrix of size $n\times n$; $T'$, $S'$, $T_1$, $T_2$ and $R$ are matrices over $\mathbb{Z}_p$; and $S$ is a matrix over $\mathbb{Z}_{p^2}$.
\end{theorem}

\begin{remark} \label{remark: zpzp2 generator matrix 2}
	\cite{z2z4_generators}
	Since $\kappa = \kappa_1 + \kappa_2$, and $\delta = \delta_1 + \delta_2$,
	where $\kappa_1$ is the dimension of the subcode $\{(\vec{u}_1, \vec{0}) \in \mathcal{C}\}$,
	and $\delta_2$ is the dimension of the subcode $\{ (\vec{0}, \vec{v}_2) \in \mathcal{C}\}$ whose codewords are all of order $p^2$,
	then the generator matrix \eqref{eq:generator matrix} can be written as:
	\begin{equation}
		\label{eq:generator matrix 2}
		\mathcal{G}_S =
		\begin{pmatrix}
			\begin{array}{cccc|ccccc}
				I_{\kappa_1} & T & T'_{p_1} & T_{p_1} & \vec{0} & \vec{0} & \vec{0} & \vec{0} & \vec{0} \\
				\vec{0} & I_{\kappa_2} & T'_{p_2} & T_{p_2} & pT_2 & pT_{\kappa_2} & \vec{0} & \vec{0} & \vec{0} \\
				\vec{0} & \vec{0} & \vec{0} & \vec{0} & pT_1 & pT'_1 & pI_{\gamma - \kappa} & \vec{0} & \vec{0} \\
				\hline
				\vec{0} & \vec{0} & S_{\delta_1} & S_p & S_{11} & S_{12} & R_1 & I_{\delta_1} & \vec{0} \\
				\vec{0} & \vec{0} & \vec{0} & \vec{0} & S_{21} & S_{22} & R_2 & R_{\delta_1} & I_{\delta_2} \\
			\end{array}
		\end{pmatrix},
	\end{equation}
	where $I_r$ is the identity matrix of size $r\times r$; $T_{p_i}, T'_{p_i}, S_{\delta_1}$ and $S_p$ are matrices over $\mathbb{Z}_p$;
	$T_1, T_2, T_{\kappa_2}, T'_1$ and $R_i$ are matrices over $\mathbb{Z}_{p^2}$ with entries in $\mathbb{Z}_p$;
	and $S_{ij}$ are matrices over $\mathbb{Z}_{p^2}$.
	Besides, $S_{\delta_1}$ and $T_{\kappa_2}$ are square matrices of full rank $\delta_1$ and $\kappa_2$, respectively.
\end{remark}

\begin{remark} \label{remark: conditions of type}
	From Theorem \ref{generator matrix of zpzp2}, there is a $\mathbb{Z}_p \mathbb{Z}_{p^2}$-additive code of type $(\alpha, \beta; \gamma, \delta; \kappa)$ if and only if
	\begin{equation} \label{eq: conditions of type}
		0 < \delta + \gamma \leqslant \beta + \kappa \quad \text{and} \quad \kappa \leqslant \min(\alpha, \gamma),
	\end{equation}
	where $\alpha, \beta, \gamma, \delta, \kappa$ are all nonnegative integers.
\end{remark}

In this paper, we will show that the conditions \eqref{eq: conditions of type} are not enough for $\mathbb{Z}_p \mathbb{Z}_{p^2}$-additive cyclic codes.

The next theorem relates the parameters of the type of a $\mathbb{Z}_p \mathbb{Z}_{p^2}$-additive cyclic code to its generator polynomials.

\begin{theorem} \label{thm: type and generator polynomials}
	\cite{Xuan-Gray-images,z2z4_generators}
	Let $\mathcal{C} = \langle (a, 0), (b, gh + pf) \rangle$ be a $\mathbb{Z}_p \mathbb{Z}_{p^2}$-additive cyclic code of type
	$(\alpha, \beta; \gamma, \delta = \delta_1 + \delta_2; \kappa = \kappa_1 + \kappa_2)$ with $fhg = x^\beta - 1$ and $\gcd(\beta, p) = 1$.
	Then
	\[ \begin{aligned}
		& \gamma = \alpha - \deg(a) + \deg(h), \quad \delta = \deg(g), \quad \kappa = \alpha - \deg(\gcd(a, b\overline{g})), \\
		& \kappa_1 = \alpha - \deg(a), \quad \kappa_2 = \deg(a) - \deg(\gcd(a, b\overline{g})),
	\end{aligned} \]
	where $\overline{g} \equiv g \pmod{p}$.
\end{theorem}

Let $\mathcal{G}$ be the generator matrix of a $\mathbb{Z}_p \mathbb{Z}_{p^2}$-additive cyclic code, then $\mathcal{G}$ and $\mathcal{G}_S$ should be permutation-equivalent,
i.e., the codes generated by $\mathcal{G}$ and $\mathcal{G}_S$ are permutation-equivalent.
Moreover, if the rows of $\mathcal{G}$ are arranged as $\mathcal{G}_S$, we call $\mathcal{G}$ is \textbf{in the form} of $\mathcal{G}_S$.

The following are some properties of the $\mathbb{Z}_p \mathbb{Z}_{p^2}$-additive cyclic codes.

\begin{lemma} \label{lemma: f'|f on Zp2}
	Let $\mathcal{C} = \langle fh + pf \rangle$ and $\mathcal{C}' = \langle f'h' + pf' \rangle$ be linear cyclic codes over $\mathbb{Z}_{p^2}$ of length $n$ with $\gcd(n, p) = 1$.
	If $\mathcal{C}\subseteq \mathcal{C}'$, then we have $f'$ divides $f$.
\end{lemma}
\begin{proof}
	By \cite[Theorem 6]{p-adic_cyclic_codes}, any ideals of $\mathbb{Z}_{p^2}[x]$ have the form $\langle f_0 + pf_1 \rangle = \langle fh + pf\rangle$
	with $f_1=f$, $f_0=fh$ and $fhg=x^n-1$.
	Since $\gcd(h, g)=1$, then $\lambda h + \mu g = 1$ for some $\lambda, \mu$.
	Then $pf = pf(\lambda h + \mu g) = p\lambda \cdot fh + \mu \cdot (pfg) \in \langle fh, pfg \rangle$.
	Thus, $\langle fh + pf \rangle \subseteq \langle fh, pfg \rangle$.
	Then $\langle fh + pf \rangle = \langle fh, pfg \rangle$ since $\langle fh + pf \rangle = \langle fh, pf \rangle$.
	Hence, $\langle fh + pf \rangle = \langle fh, pfg \rangle$.
	Any cyclic code over $\mathbb{Z}_{p^2}$ can be written as
	\begin{equation} \label{eq: <fh+pf> = <fh+pfg>}
		\mathcal{C} = \langle fh + pf \rangle = \langle fh, pfg \rangle
	\end{equation}
	with $fhg = x^n-1$ and $\gcd(n, p) = 1$.
	Similar to the proof of \cite[Theorem 3]{kernel_of_4-ary_code}, $f'$ divides $f$.
\end{proof}

\begin{proposition} \label{prop: f|f_k}
	Let $\mathcal{C} = \langle (a, 0), (b, fh+pf) \rangle$ and $\mathcal{C}' = \langle (a', 0), (b', f'h'+pf') \rangle$ be $\mathbb{Z}_p\mathbb{Z}_{p^2}$-additive cyclic
	codes with $\mathcal{C} \subseteq \mathcal{C}'$, then $f'$ divides $f$ and $\gcd(a', b')$ divides $\gcd(a, b)$.
\end{proposition}
\begin{proof}
	It's easy to check that $\mathcal{C}_Y = \langle (fh+pf) \rangle \subseteq \langle (f'h'+pf')\rangle = \mathcal{C}'_Y$ since $\mathcal{C}\subseteq \mathcal{C}'$.
	By Lemma \ref{lemma: f'|f on Zp2}, $f' \mid f$.
	Obviously, $\mathcal{C}_X = \langle \gcd(a, b) \rangle$ and $\mathcal{C}'_X = \langle \gcd(a', b ') \rangle$.
	Thus, $\gcd(a', b ')$ divides $\gcd(a, b)$ since $\mathcal{C}_X \subseteq \mathcal{C}'_X$.
\end{proof}

\begin{proposition} \label{prop: intersection of C_1 and C_2 over Zp2}
	Let $\mathcal{C}_i = \langle fh_i + pf \rangle$ be cyclic codes over $\mathbb{Z}_{p^2}$, $i=1,2$, then
	\[ \mathcal{C}_1 \cap \mathcal{C}_2 = \langle f \cdot \mathrm{lcm}(h_1, h_2), pf\gcd(g_1, g_2)\rangle, \]
	where $\mathrm{lcm}(h_1, k_2)$ means the least common multiple of $h_1$ and $h_2$.
\end{proposition}
\begin{proof}
	Let $\mathcal{D} = \langle f \cdot \lcm(h_1, h_2), pf\gcd(g_1, g_2)\rangle$.
	It's easy to check that $fh_1\mid f \cdot \lcm(h_1, h_2)$, then $f \cdot \lcm(h_1, h_2) \in \mathcal{C}_1$.
	Since $\gcd(g_1, h_1) = 1$, then $\lambda h_1 + \mu g_1 = 1$ for some $\lambda, \mu\in \mathbb{Z}_{p^2}[x]$.
	Then we have \[ pf\gcd(g_1, g_2) = pf\gcd(g_1, g_2)\left( \lambda h_1 + \mu g_1 \right) = p\lambda\gcd(g_1, g_2)(fh_1) + \mu\gcd(g_1, g_2) (pfg_1), \]
	which means $pf\gcd(g_1, g_2) \in \mathcal{C}_1$ and $D\subseteq \mathcal{C}_1$.
	Similarly, we have $\mathcal{D}\subseteq \mathcal{C}_2$ and thus $D\subseteq \mathcal{C}_1\cap \mathcal{C}_2$.
	
	Without loss of generality, let $\mathcal{C}_3 =\mathcal{C}_1 \cap \mathcal{C}_2 = \langle f'h', pf'g'\rangle$ with $f'h'g' = x^n-1$,
	then $f\mid f'$ and $fh_i\mid f'h'$ since $\mathcal{C}_3\subseteq \mathcal{C}_i$ for $i=1,2$.
	Thus, $f\cdot \lcm(h_1, h_2)\mid f'h'$.
	Note that $\lcm(h_1, h_2)$ and $\gcd(g_1, g_2)$ are coprime and $\lcm(h_1, h_2) \cdot \gcd(g_1, g_2) = (x^n-1)/f$.
	As for $pf'g'$, similar to the above discussion, we can prove that it is also in $\mathcal{D}$.
	Therefore, $\mathcal{D} = \mathcal{C}_1 \cap \mathcal{C}_2$.
\end{proof}

The generator polynomials of $\mathcal{C}$ are not unique.

\begin{lemma} \label{3 generators of zpzp2 cyclic codes}
	\cite{binary_images_of_z2z4_cyclic,Xuan-Gray-images}
	Let $\mathcal{C} = \langle (a, 0), (b, fh + pf)\rangle$ be a $\mathbb{Z}_p\mathbb{Z}_{p^2}$-additive cyclic code with $fhg = x^\beta-1$.
	Then $\mathcal{C}$ can be also generated by $(a, 0)$, $(b\overline{g}, pfg)$ and $(b', fh)$, where $b' = b - \overline{\mu}b\overline{g}$ and $\lambda h + \mu g = 1$.
\end{lemma}

Let $\mathcal{C}_p = \{ \vec{c}| p \cdot \vec{c} = \vec{0}, \vec{c} \in \mathcal{C} \}$ be the subcode of $\mathcal{C}$, then we have

\begin{lemma}
	\label{generators of C_p}
	\cite{Xuan-Gray-images}
	Let $\mathcal{C} = \langle (a, 0), (b, fh + pf)\rangle$ be a $\mathbb{Z}_p\mathbb{Z}_{p^2}$-additive cyclic code with $fhg = x^\beta-1$, then
	$\mathcal{C}_p$ is generated by $(a , 0), (b  \overline{g }, pf  g )$ and $(0, pf  h )$.
\end{lemma}

\subsection{Gray Map and Gray Image}\label{subsec:Gray map}

In \cite{LS_Gray_map}, the classical Gray map $\phi: \mathbb{Z}_{p^2} \rightarrow \mathbb{Z}_p^p$ is defined as
\[ \phi(\theta) = \theta_0 (1, 1, \cdots, 1) + \theta_1(0, 1, \cdots, p-1), \]
where $\theta = \theta_0 p + \theta_1 \in \mathbb{Z}_{p^2}$ and $\phi(\theta)$ is a vector of length $p$.
Then let $n = \alpha + p\beta$ and $\Phi: \mathbb{Z}_p^{\alpha} \times \mathbb{Z}_{p^2}^{\beta} \rightarrow \mathbb{Z}_p^n$ be an extension of the Gray map $\phi$,
which is defined as \[ \Phi(\vec{x}, \vec{y}) = (\vec{x}, \phi(y_1), \phi(y_2), \cdots, \phi(y_{\beta})) \]
for any $\vec{x} \in \mathbb{Z}_p^{\alpha}$, and $\vec{y} = (y_1, y_2, \cdots, y_{\beta}) \in \mathbb{Z}_{p^2}^{\beta}$.
Besides, in the following sections, the Gray map refers to $\Phi$ instead of $\phi$.
The $p$-ary code $C = \Phi(\mathcal{C})$ is called a $\mathbb{Z}_p \mathbb{Z}_{p^2}$-\textit{linear code}
or the \textit{$p$-ary image} of a $\mathbb{Z}_p\mathbb{Z}_{p^2}$-additive code.
And we denote $\langle C \rangle$ the \textit{linear span} of the codewords of $C$.

In \cite{wan1_Z4-code}, the Gray map $\Phi$ defined for quaternary codes satisfies $\Phi(\vec{u} + \vec{v}) = \Phi(\vec{u}) + \Phi(\vec{v}) + \Phi(2 \vec{u} \ast \vec{v})$,
where $\ast$ denotes the \textbf{componentwise product} of two vectors.

Recall some results from \cite{Xuan-Gray-images}.

\begin{lemma} \label{phi(u+v) = phi(u)+phi(v)+...}
	\cite{Xuan-Gray-images,wsk_z3z9_linear_rank_kernel}
	Let $\vec{u} = (u_1, \cdots, u_{\alpha + \beta}), \vec{v} = (v_1, \cdots, v_{\alpha + \beta}) \in \mathbb{Z}_p^{\alpha} \times \mathbb{Z}_{p^2}^{\beta}$, then
	\begin{equation} \label{eq: Phi(u+v)=Phi(u)+Phi(v)+Phi(pP(u,v))}
		\Phi(\vec{u} + \vec{v}) = \Phi(\vec{u}) + \Phi(\vec{v}) + \Phi(pP(\vec{u}, \vec{v})),
	\end{equation}
	where $P(\vec{u}, \vec{v}) = P(u_1, v_1), \cdots, P(u_{\alpha+\beta}, v_{\alpha+\beta})$ and
	\begin{equation}\label{eq: values of P(u,v)}
		P(a+pc, b+pd) = P(a,b) =
		\begin{cases}
			1, & a + b \geqslant p, \\
			0, & a + b < p, \\
		\end{cases}
	\end{equation}
	with $a+pb, b+pd \in \mathbb{Z}_{p^2}$ and $a,b,c,d\in \mathbb{Z}_p$.
	Moreover, the values of $P(u_i, v_i)$ equal to
	\begin{equation}\label{eq: P(u,v)}
		f(u_i, v_i) = \sum_{k=1}^{p-1} \sum_{b=p-k}^{p-1} \frac{\prod_{m=0}^{p-1} (u_i-m)(v_i-m)}{(u_i-k)(v_i-b)}.
	\end{equation}
	In particular, if the order of $\vec{u}$ or $\vec{v}$ is $p$, then $\Phi(\vec{u} + \vec{v}) = \Phi(\vec{u}) + \Phi(\vec{v})$.
\end{lemma}

\begin{remark} \label{remark: values of P(a,b)}
	It's clear that $P(u_i, v_i) = f(u_i, v_i)$ for all $u_i, v_i \in \mathbb{Z}_{p^2}$.
	In \cite[Section IV]{LS_Gray_map}, the authors also defined a characteristic function that satisfies \eqref{eq: values of P(u,v)}.
	Although the function \eqref{eq: values of P(u,v)} is very important, any polynomial $f(x, y) \in \mathbb{Z}_{p}[x_1,x_2]$ that satisfies \eqref{eq: values of P(u,v)}
	can be chosen as $P(x,y)$. In the next sections, let $m$ be the smallest degree of such polynomials.
\end{remark}

\begin{lemma} \label{lemma: Phi(C) is linear iff pP(u,v) in C}
	\cite{Xuan-Gray-images}
	Let $\mathcal{C}$ be a $\mathbb{Z}_p\mathbb{Z}_{p^2}$-additive code.
	The $\mathbb{Z}_p\mathbb{Z}_{p^2}$-linear code $C = \Phi(\mathcal{C})$ is linear if and only if $pP(\vec{u},\vec{v}) \in \mathcal{C}$ for all $\vec{u},\vec{v} \in \mathcal{C}$.
\end{lemma}

Let $c_j = \left( c\otimes c \otimes \cdots \otimes c \right) $ be the polynomial whose roots are the products $\xi^{i_1} \xi^{i_2} \cdots \xi^{i_j}$
such that $\xi^{i_1}, \xi^{i_2}, \cdots, \xi^{i_j}$ are roots of $c$, the divisor of $x^n-1 \in \mathbb{Z}_p[x]$ with $\gcd(n,p) = 1$.
Then $c_j$ is called the $j_{th}$ \textbf{circle product} of $c$.

\begin{theorem} \label{thm: linearity of images of ZpZp2 cyclic codes}
	\cite{Xuan-Gray-images}
	Let $\mathcal{C} = \langle (a, 0), (b, fh + pf)\rangle$ be a $\mathbb{Z}_p\mathbb{Z}_{p^2}$-additive cyclic code with $fhg = x^\beta-1$.
	Assume that $(x-1) \mid g$ and $\Phi(p \vec{u}_1 \ast \cdots \ast \vec{u}_m) \in \langle \Phi(\mathcal{C}) \rangle$ for any $\vec{u}_i \in C$, $1 \leqslant i \leqslant m$.
	Then $\Phi(\mathcal{C})$ is linear if $\gcd\left( f', \overline{g_m} \right) = 1$ with $f' = \overline{f}a / \gcd(a, \overline{b}g)$,
	where $\overline{g_m} = \left( \overline{g}\otimes\cdots\otimes\overline{g} \right)$ is the $m_{th}$ circle product of $\overline{g}$ and $m=\deg(P(x_1, x_2))$.
\end{theorem}

\subsection{Kernels and Ranks}\label{subsec:kernel and rank}

Let $\mathcal{C}$ be a $\mathbb{Z}_p \mathbb{Z}_{p^2}$-additive code, then the \textit{kernel} of $\Phi(\mathcal{C})$ is defined as
\[ \ker(\Phi(\mathcal{C})) = \{ \vec{v} \in \mathbb{Z}_p^{\alpha + p\beta} \mid \vec{v} + \Phi(\mathcal{C}) = \Phi(\mathcal{C}) \}. \]
And, the \textit{kernel} of $\mathcal{C}$ is defined as \[ \mathcal{K}(\mathcal{C}) = \Phi^{-1} (\ker(\Phi(\mathcal{C}))). \]
For convenience, let $ker(\mathcal{C})$ be the dimension of $\mathcal{K}(\mathcal{C})$.

Let rank$(\Phi(\mathcal{C})) = \dim(\langle \Phi (\mathcal{C}) \rangle)$ and
\[ \mathcal{R}(\mathcal{C}) = \{ \vec{v}\in \mathbb{Z}_p^{\alpha} \times \mathbb{Z}_{p^2}^{\beta} \mid \Phi(\vec{v}) \in \langle \Phi(\mathcal{C}) \rangle \}.\]

In this section, we will see that if the code $\mathcal{C}$ is a cyclic code, then $\mathcal{R}(\mathcal{C})$
and $\mathcal{K}(\mathcal{C})$ are both cyclic codes whether the alphabet is $\mathbb{Z}_{p^2}$ or $\mathbb{Z}_p\mathbb{Z}_{p^2}$.

By \cite[Lemma 10]{wsk_z3z9_linear_rank_kernel}, Lemma \ref{lemma: Phi(C) is linear iff pP(u,v) in C} and the definition of $\mathcal{K}(\mathcal{C})$, we have

\begin{proposition} \label{prop: ker(Phi(C)) and K(C)}
	Let $\mathcal{C}$ be a $\mathbb{Z}_p\mathbb{Z}_{p^2}$-additive cyclic code and $C = \Phi(\mathcal{C})$, then
	\[ \ker(\Phi(\mathcal{C})) = \{ \Phi(\vec{u}) \mid \vec{u}\in \mathcal{C} \ \textnormal{and} \ pP(\vec{u}, \vec{v}) \in \mathcal{C}, \forall \vec{v} \in \mathcal{C} \}, \]
	and
	\[ \mathcal{K}(\mathcal{C}) = \{ \vec{u} \in \mathcal{C} \mid pP(\vec{u}, \vec{v}) \in \mathcal{C}, \forall \vec{v} \in \mathcal{C} \}. \]
\end{proposition}

\begin{remark}
	Proposition \ref{prop: ker(Phi(C)) and K(C)} also holds for $\mathbb{Z}_{p^2}$-cyclic codes.
\end{remark}

\section{Kernel of $\mathbb{Z}_p\mathbb{Z}_{p^2}$-Additive Cyclic Codes}\label{sec:kernel}

In this section, we will study the kernel of $\mathbb{Z}_p\mathbb{Z}_{p^2}$-additive cyclic codes.
We will prove that $\mathcal{K}(\mathcal{C})$ is cyclic if $\mathcal{C}$ is a $\mathbb{Z}_p\mathbb{Z}_{p^2}$-additive cyclic code.
For convenience, let $ker(C)$ or $ker(\mathcal{C})$ be the dimension of the kernel of $C=\Phi(\mathcal{C})$.

\begin{lemma} \label{K(C)_Y < K(C_Y)}
	Let $\mathcal{C}$ be a  $\mathbb{Z}_p\mathbb{Z}_{p^2}$-additive code, then $\mathcal{K}(\mathcal{C})_Y \subseteq \mathcal{K}(\mathcal{C}_Y)$.
\end{lemma}
\begin{proof}
	Let $\mathcal{C}$ be a $\mathbb{Z}_p\mathbb{Z}_{p^2}$-additive cyclic code and $\vec{u} = (\vec{u}_1 ,  \vec{u}_2) \in \mathcal{C}$ is a codeword.
	Then $\vec{u} \in \mathcal{K}(\mathcal{C})$ if and only if $pP(\vec{u}, \vec{v}) \in \mathcal{C}$ for $\vec{v} = (\vec{v}_1 ,  \vec{v}_2) \in \mathcal{C}$.
	It's easy to check that for all $\vec{v} \in \mathcal{C}$, we have
	\[ pP(\vec{u}, \vec{v}) = (\vec{0}, pP(\vec{u}_2, \vec{v}_2)). \]
	Thus, if $\vec{u} = (\vec{u}_1 ,  \vec{u}_2) \in \mathcal{K}(\mathcal{C})$, then $\vec{u}_2 \in \mathcal{K}(\mathcal{C})_Y$ and $\vec{u}_2 \in \mathcal{K}(\mathcal{C}_Y)$.
\end{proof}

\begin{proposition} \label{K(C) < C_X * K(C_Y)}
	Let $\mathcal{C}$ be a $\mathbb{Z}_p\mathbb{Z}_{p^2}$-additive code, then $\mathcal{K}(\mathcal{C}) \subseteq \mathcal{C}_X \times \mathcal{K}(\mathcal{C}_Y)$.
\end{proposition}
\begin{proof}
	Let $\vec{u} = (\vec{u}_1 ,  \vec{u}_2) \in \mathcal{K}(\mathcal{C})$, then for all $\vec{v} = (\vec{v}_1 ,  \vec{v}_2) \in \mathcal{C}$,
	we have $pP(\vec{u}, \vec{v}) \in \mathcal{C}$ since $\vec{u}$ is in the kernel.
	By the proof of Lemma \ref{K(C)_Y < K(C_Y)}, we know $\vec{u}_1 \in \mathcal{C}_X$ and $pP(\vec{u}_2, \vec{v}_2) \in \mathcal{C}_Y$,
	for all $\vec{v}_2 \in \mathcal{C}_Y$ which gives $\vec{u}_2 \in \mathcal{K}(\mathcal{C}_Y)$ and $\vec{u} \in \mathcal{C}_X \times \mathcal{K}(\mathcal{C}_Y)$.
\end{proof}

\begin{proposition} \label{K(C) = C_X * K(C_Y)}
	If $\mathcal{C}$ is a separable $\mathbb{Z}_p\mathbb{Z}_{p^2}$-additive code, then $\mathcal{K}(\mathcal{C}) = \mathcal{C}_X \times \mathcal{K}(\mathcal{C}_Y)$.
\end{proposition}
\begin{proof}
	Let $\vec{u} \in \mathcal{C}_X \times \mathcal{K}(\mathcal{C}_Y)$, then $pP(\vec{u}, \vec{v}) = (\vec{0} ,  pP(\vec{u}_2, \vec{v}_2))$ for all
	$\vec{v} \in \mathcal{C}$.
	It's also easy to check that $\vec{u}_2 \in \mathcal{K}(\mathcal{C}_Y)$ and $pP(\vec{u}_2, \vec{v}_2) \in \mathcal{C}_Y$.
	Note that $\vec{0} \in \mathcal{C}_X$.
	Thus, we have \[ pP(\vec{u}, \vec{v}) = (\vec{0} ,  pP(\vec{u}_2, \vec{v}_2) )\in \mathcal{C} = \mathcal{C}_X \times \mathcal{C}_Y, \]
	for all $\vec{v} \in \mathcal{C}$ since $\mathcal{C}$ is separable.
	Hence, $\vec{u} \in \mathcal{K}(\mathcal{C})$ and $\mathcal{C}_X \times \mathcal{K}(\mathcal{C}_Y) \subseteq \mathcal{K}(\mathcal{C})$.
	By Proposition \ref{K(C) < C_X * K(C_Y)}, we know $\mathcal{K}(\mathcal{C}) = \mathcal{C}_X \times \mathcal{K}(\mathcal{C}_Y)$
	if $\mathcal{C}$ is separable.
\end{proof}

From the generator matrix given in \eqref{eq:generator matrix 2}, the code $\mathcal{C}_Y$ has a generator matrix of the following form
\begin{equation}
	\label{eq:generator matrix of C_Y}
	\begin{pmatrix}
		\begin{array}{ccc}
			pT_2 & \vec{0} & \vec{0} \\
			pT_1 & pI_{\gamma - \kappa} & \vec{0} \\
			S & R & I_\delta
		\end{array}
	\end{pmatrix}.
\end{equation}
Obviously, $\mathcal{C}_Y$ is a linear code over $\mathbb{Z}_{p^2}$.
By \cite[Sec.2]{p-adic_cyclic_codes}, $\mathcal{C}_Y$ has type $p^{2\delta} p^{\gamma-\kappa_1}$.
The minimum value for the dimension of $\ker(\phi(\mathcal{C}_Y))$ is $\delta + \gamma - \kappa_1$.

\begin{theorem} \label{thm: dimension of the kernel}
	Let $\mathcal{C}$ be a $\mathbb{Z}_p\mathbb{Z}_{p^2}$-additive code with the generator matrix in the form of \eqref{eq:generator matrix 2} and let $\mathcal{C}'$ be the code
	generated by the matrix $G'$ in \eqref{eq:generator matrix of C'}, where
	\begin{equation} \label{eq:generator matrix of C'}
		G' =
		\left(  \begin{array}{ccc|ccc}
			\vec{0} & \vec{0} & \vec{0} & p T_1 & p I_{\gamma - \kappa} & \vec{0} \\
			\vec{0} & \vec{0} & S' & S & R & I_{\delta} \\
		\end{array} \right).
	\end{equation}
	Then $\dim(\ker(\Phi(\mathcal{C}))) = \kappa_1 + \kappa_2 + \dim(\ker(\phi(\mathcal{C}'_Y)))$.
\end{theorem}
\begin{proof}
	Let $\vec{u}_i = (u_i ,  u'_i)$ be the first $\gamma$ rows and $\vec{v}_j = (v_j ,  v'_j)$ be the last $\delta$ rows of $G$, respectively,
	where $i=1, 2, \cdots, \gamma$, $j=1, 2, \cdots, \delta$ and $G$ is in the form of \eqref{eq:generator matrix 2}.
	Since $\gamma \geqslant \kappa_1 + \kappa_2 = \kappa$, then let $\overline{\mathcal{C}} = \langle \vec{u}_1, \cdots, \vec{u}_{\kappa_1 + \kappa_2} \rangle$
	and $\mathcal{C}' = \langle \vec{u}_{\kappa_1 + \kappa_2 + 1}, \cdots, \vec{u}_\gamma, \vec{v}_1, \cdots, \vec{v}_\delta \rangle$.
	
	By Proposition \ref{prop: ker(Phi(C)) and K(C)}, $\vec{u} \in \mathcal{K}(\mathcal{C})$ if and only if $pP(\vec{u}, \vec{v}) \in \mathcal{C}$ for all
	$\vec{v} \in \mathcal{C}$.
	Then we have the following cases:
	\begin{enumerate}
		\item[(1)] If $\vec{u} \in \overline{\mathcal{C}}$, we have $pP(\vec{u}, \vec{v}) = \vec{0}$, for all $\vec{v} \in \mathcal{C}$
		since all codewords in $\overline{\mathcal{C}}$ are of order $p$.
		Hence, $\vec{u} \in \mathcal{K}(\mathcal{C})$, $\overline{\mathcal{C}} \subseteq \mathcal{K}(\mathcal{C})$,
		and $\dim(\ker(\overline{\mathcal{C}})) = \kappa_1 + \kappa_2$.
		
		\item[(2)] Let $\vec{u} = (u, u') \in \mathcal{C}'$. By the above discussion, $pP(\vec{u}, \vec{v}) = \vec{0}$ for all $\vec{v} \in \overline{\mathcal{C}}$.
		Then $\vec{u} \in \mathcal{K}(\mathcal{C})$ if and only if $pP(\vec{u}, \vec{v}) \in \mathcal{C}'$ for all $\vec{v} \in \mathcal{C}'$;
		that is $\vec{v} \in \mathcal{K}(\mathcal{C}')$.
		Note that $pP(\vec{u}, \vec{v}) = (\vec{0}, pP(u', v')) \in \mathcal{C}'$ if and only if $pP(u', v') \in \mathcal{C}'_Y$.
		Hence, \[ \dim(\ker(\Phi(\mathcal{C}'))) = \dim(\ker(\phi(\mathcal{C}'_Y))). \]
	\end{enumerate}
	Therefore, $\dim(\ker(\Phi(\mathcal{C}))) = \kappa_1 + \kappa_2 + \dim(\ker(\phi(\mathcal{C}'_Y)))$ and the proof is completed.
\end{proof}

From Remark \ref{remark: conditions of type}, there exists a $\mathbb{Z}_p \mathbb{Z}_{p^2}$-additive code for all possible values of the kernel for a given type.
But this theorem doesn't hold for $\mathbb{Z}_p \mathbb{Z}_{p^2}$-additive cyclic code.

\begin{example} \label{exam: z3z9-kernel-1}
	By Remark \ref{remark: conditions of type}, there exists a $\mathbb{Z}_3 \mathbb{Z}_{9}$-additive code $\mathcal{C}$ of type $(\alpha, 13; 2, 4; 1)$.
	The possible dimensions of the kernel are $6,7,8,9,10$.
	Let $\mathcal{C} = \langle (a, 0), (b, fh+3f) \rangle$ be a $\mathbb{Z}_3 \mathbb{Z}_9$-additive cyclic code of type $(2, 13; 2, 4; 1)$, where $fhg = x^{13}-1$.
	By Magma, factoring $x^{13}-1$ over $\mathbb{Z}_9$, we have \[ x^{13} - 1 = (x + 8)(x^3 + 6x^2 + 2x + 8)(x^3 + 7x^2 + 3x + 8)(x^3 + 4x^2 + 7x + 8)(x^3 + 2x^2 + 5x + 8). \]
	Since $\gamma = \deg(h) = 2$, then there doesn't exist such a $\mathbb{Z}_3 \mathbb{Z}_9$-additive cyclic code.
\end{example}

\begin{theorem} \label{K(C) is cyclic}
	If $\mathcal{C}$ is a $\mathbb{Z}_p\mathbb{Z}_{p^2}$-additive cyclic code, then $\mathcal{K}(\mathcal{C})$ is also a $\mathbb{Z}_p\mathbb{Z}_{p^2}$-additive cyclic code.
\end{theorem}
\begin{proof}
	It's easy to check that $\mathcal{K}(\mathcal{C})$ is a $\mathbb{Z}_p\mathbb{Z}_{p^2}$-additive code by Proposition \ref{prop: ker(Phi(C)) and K(C)}.
	It's sufficient to show that if $\vec{u} = (u, u') \in \mathcal{K}(\mathcal{C})$, then $\pi(\vec{u}) \in \mathcal{K}(\mathcal{C})$, which means
	$pP(\pi(\vec{u}), \vec{v}) \in \mathcal{C}$ for all $\vec{v} \in \mathcal{C}$.
	
	Let $\vec{u} \in \mathcal{K}(\mathcal{C})$ and $\vec{v} \in \mathcal{C}$.
	It's easy to check that $pP(\pi(\vec{u}), \vec{v}) = \pi(pP(\vec{u}, \pi^{-1}(\vec{v})))$.
	Then $pP(\vec{u}, \pi^{-1}(\vec{v})) \in \mathcal{C}$ by Proposition \ref{prop: ker(Phi(C)) and K(C)}.
	Since $\mathcal{C}$ is cyclic, then $\pi(pP(\vec{u}, \pi^{-1}(\vec{v}))) \in \mathcal{C}$.
	Thus, $pP(\pi(\vec{u}), \vec{v}) \in \mathcal{C}$ for all $\vec{v} \in \mathcal{C}$ and $\pi(\vec{u}) \in \mathcal{K}(\mathcal{C})$.
\end{proof}

\begin{lemma} \label{lemma: cyclic subcodes of cyclic codes over zpzp2}
	Let $\mathcal{C} = \langle (a,0), (b, fh+pf) \rangle$ be a $\mathbb{Z}_p\mathbb{Z}_{p^2}$-additive cyclic code, where $fhg = x^\beta-1$ and $\gcd(\beta,p) = 1$.
	Let $\langle (a, 0), (b_k, fhk + pf) \rangle \subseteq \mathcal{C}$, for $k \mid g$.
	Then $b_k \equiv \overline{k}b + (1-\overline{k}) \overline{\mu} b \overline{g} \pmod{a}$ where $\mu$ satisfies  $\lambda h + \mu g = 1$.
\end{lemma}
\begin{proof}
	First, we will prove that $\mathcal{C} = \langle (a, 0), (b', fh), (\overline{\mu}b\overline{g}, pf) \rangle$,
	where $b' = b-\overline{\mu}b\overline{g}$ and $\lambda h + \mu g = 1$. Let $\mathcal{D} = \langle (a, 0), (b', fh), (\overline{\mu}b\overline{g}, pf) \rangle$, and we shall prove that $\mathcal{C} = \mathcal{D}$.
	It's easy to check that $(0, pfh) = p(b, fh + pf)\in \mathcal{C}$ and $(b\overline{g}, pfg) = g \star (b, fh + pf) \in \mathcal{C}$, then
	\[ \lambda \star (0, pfh) + \mu \star (b\overline{g}, pfg) = (\overline{\mu} b \overline{g}, pf) \in \mathcal{C}. \]
	By Lemma \ref{3 generators of zpzp2 cyclic codes}, we have $\mathcal{D} \subseteq \mathcal{C}$.
	Since $(b, fh + pf) = (b', fh) + (\overline{\mu}b\overline{g}, pf) \in \mathcal{D}$, then $\mathcal{C} \subseteq \mathcal{D}$, which means
	\begin{equation} \label{eq: 3 generators of zpzp2 cyclic codes}
		\mathcal{C} = \langle (a, 0), (b', fh), (\overline{\mu}b\overline{g}, pf) \rangle.
	\end{equation}
	Then, since $(b_k, fhk + pf)\in \mathcal{C}$, we have
	\[ (b_k, fhk + pf) = c_1 \star (a, 0) + c_2 \star (b', fh) + c_3 \star (\overline{\mu}b\overline{g}, pf). \]
	Thus, we get $c_2 = k$, $c_3 = 1$ and $b_k \equiv \overline{k}b + (1-\overline{k}) \overline{\mu} b \overline{g} \pmod{a}$.
\end{proof}

\begin{proposition} \label{prop: K(C)=fhk+pfg/k over Zp2}
	Let $\mathcal{C} = \langle fh + pf \rangle$ be a linear cyclic code over $\mathbb{Z}_{p^2}$ of length $n$ with $\gcd(n, p) = 1$.
	Then $\mathcal{K}(\mathcal{C})$ is also cyclic and $\langle pf \rangle \subseteq \mathcal{K}(\mathcal{C})$.
	Moreover, $\mathcal{K}(\mathcal{C}) = \langle fhk, pfg_k\rangle$ where $g_k = g/k$ and $k$ divides $g$.
\end{proposition}
\begin{proof}
	By the proof of \cite[Theorem 5]{kernel_of_4-ary_code}, we can get $\mathcal{K}(\mathcal{C})$ is cyclic.
	By Lemma \ref{generators of C_p}, the subcode whose codewords are of order $p$ is just $\langle pfg, pfh \rangle$.
	Since $g$ and $h$ are coprime, then $\langle pfg, pfh \rangle = \langle pf \rangle$.
	According to Lemma \ref{phi(u+v) = phi(u)+phi(v)+...} and Proposition \ref{prop: ker(Phi(C)) and K(C)},
	$\langle pf \rangle \subseteq \mathcal{K}(\mathcal{C})$. Let $\mathcal{K}(\mathcal{C}) = \langle f'h' + pf' \rangle$ with $f'h'g' = x^n-1$.
	Since \[ \langle pf \rangle = \langle x^n-1, pf\rangle = \langle f\cdot \frac{x^n-1}{f}, pf\rangle \subseteq \langle f'h', p f'g' \rangle, \]
	then by Proposition \ref{prop: f|f_k}, $f'\mid f$.
	Moreover, $f\mid f'$ since $\mathcal{K}(\mathcal{C}) \subseteq \mathcal{C}$.
	Thus, $f' = f$ and $\mathcal{K}(\mathcal{C}) = \langle fh' + pf' \rangle$.
	The proof of the rest claim is similar to that of \cite[Theorem 9]{kernel_of_4-ary_code}.
\end{proof}

\begin{theorem} \label{thm: properties of generators of K(C)}
	Let $\mathcal{C} = \langle (a,0), (b, fh+pf) \rangle$ be a $\mathbb{Z}_p\mathbb{Z}_{p^2}$-additive cyclic code, where $fhg = x^\beta-1$ and $\gcd(\beta,p) = 1$.
	Then, $\mathcal{K}(\mathcal{C}) = \langle (a, 0), (b_k, fhk + pf) \rangle$,
	where $k \mid g$, $b_k \equiv \overline{k}b + (1-\overline{k}) \overline{\mu} b \overline{g} \pmod{a}$ and $\mu$ satisfies $\lambda h + \mu g = 1$.
\end{theorem}
\begin{proof}
	By Theorem \ref{K(C) is cyclic}, let $\mathcal{K}(\mathcal{C}) = \langle (a_k, 0), (b_k ,  f_k h_k + pf_k) \rangle$, where $f_k g_k h_k = x^\beta-1$.
	It's easy to check that $\mathcal{C}_p \subseteq \mathcal{K}(\mathcal{C}) \subseteq \mathcal{C}$.
	By Lemma \ref{generators of C_p}, $(a, 0) \in \mathcal{K}(\mathcal{C})$, then we have $a = a_k$.
	Besides, $(\mathcal{C}_p)_Y = \langle pfg, pfh \rangle = \langle pf \rangle \subseteq \langle f_k h_k + pf_k \rangle = (\mathcal{K}(\mathcal{C}))_Y$
	since $g$ and $h$ are coprime.
	According to Proposition \ref{prop: f|f_k}, we have $f = f_k$.
	
	By Proposition \ref{prop: K(C)=fhk+pfg/k over Zp2}, we get that $h_k = hk$ with $k$ a divisor of $g$
	since $(\mathcal{K}(\mathcal{C}))_Y = \langle f h_k + pf \rangle \subseteq \mathcal{C}_Y$.
	Moreover, just as \eqref{eq: 3 generators of zpzp2 cyclic codes} shows, $(b', fh), (\overline{\mu} b \overline{g}, pf)\in \mathcal{C}$,
	where $b' = b - \overline{\mu} b \overline{g}$ and $\lambda h + \mu g = 1$.
	Thus, we get $b_k = \overline{k}b' + \overline{\mu} b \overline{g} \equiv \overline{k}b + (1-\overline{k}) \overline{\mu} b \overline{g} \pmod{a}$.
\end{proof}

\begin{remark} \label{remark: generators of K(C)}
	Let $\mathcal{C}'$ be a maximal cyclic subcode of $\mathbb{Z}_p\mathbb{Z}_{p^2}$-additive cyclic code $\mathcal{C}$
	whose Gray image $\Phi(\mathcal{C}')$ is a linear subcode of $\Phi(\mathcal{C})$.
	Then we also have $\mathcal{C}_p \subseteq \mathcal{C}' \subseteq \mathcal{C}$.
	Thus, by an argument analogous to that of Theorem \ref{thm: properties of generators of K(C)},
	we can get the statement in Theorem \ref{thm: properties of generators of K(C)} also holds for $\mathcal{C}'$.
\end{remark}

\begin{proposition} \label{prop: intersection of C_1 and C_2}
	Let $\mathcal{C} = \langle (a, 0), (b, fh + pf)\rangle$ be a $\mathbb{Z}_p\mathbb{Z}_{p^2}$-additive cyclic code.
	Assume that $\mathcal{C}_1 = \langle (a, 0), (b_1, fhk_1 + pf)\rangle$ and $\mathcal{C}_2 = \langle (a, 0), (b_2, fhk_2 + pf)\rangle$ are $\mathbb{Z}_p\mathbb{Z}_{p^2}$-additive
	maximal subcodes of $\mathcal{C}$ whose images under the Gray map are linear subcodes of $\Phi(\mathcal{C})$, then
	\[ \mathcal{C}_1 \cap \mathcal{C}_2 = \langle (a, 0), (b', fhk' + pf)\rangle, \]
	where $k' = \lcm(k_1, k_2)$ and $b' \equiv \overline{k'} b + (1-\overline{k'})\overline{\mu} b \overline{g} \pmod{a}$,
	with $\lambda h + \mu g = 1$.
\end{proposition}
\begin{proof}
	By Proposition \ref{prop: intersection of C_1 and C_2 over Zp2},
	we have $(\mathcal{C}_1)_Y \cap (\mathcal{C}_2)_Y = \langle f \cdot \lcm(hk_1, hk_2), pf\rangle = \langle fhk' + pf\rangle$ with $k' = \lcm(k_1, k_2)$.
	Since $\langle (a, 0) \rangle \subseteq \mathcal{C}_1 \cap \mathcal{C}_2$, then $\mathcal{C}_1 \cap \mathcal{C}_2 = \langle(a, 0), (b', fhk' + pf)\rangle$.
	By Lemma \ref{lemma: cyclic subcodes of cyclic codes over zpzp2}, we can get $b' \equiv \overline{k'} b + (1-\overline{k'})\overline{\mu} b \overline{g} \pmod{a}$,
	with $\lambda h + \mu g = 1$.
\end{proof}

\begin{lemma} \label{lemma: K(C) is the intersection of maximal subcodes}
	Let $\mathcal{C}$ be a $\mathbb{Z}_p\mathbb{Z}_{p^2}$-additive cyclic code.
	If $\mathcal{D}$ is a maximal cyclic subcode with linear $p$-ary image, then $\mathcal{K}(\mathcal{C}) \subseteq \mathcal{D}$.
	Let $\mathcal{C}_i$ be all the maximal subcodes of $\mathcal{C}$ such that $\Phi(\mathcal{C}_i)$ is a linear subcode of $\Phi(\mathcal{C})$ for $1=1, \cdots, s$, then
	\[ \mathcal{K}(\mathcal{C}) = \bigcap_{i=1}^s \mathcal{C}_i = \bigcap_{i=1}^s \mathcal{D}_i, \]
	where $\mathcal{D}_i$ is a maximal cyclic subcode of $\mathcal{C}$ with linear $p$-ary image and $\mathcal{D}_i \subseteq \mathcal{C}_i$.
\end{lemma}
\begin{proof}
	If $\mathcal{K}(\mathcal{C}) \nsubseteq \mathcal{D}$, then consider the code $\mathcal{D}'$ generated by
	$\mathcal{K}(\mathcal{C}) \cup \mathcal{D} \cup \{ pP(\vec{u}, \vec{v}) | \vec{u}, \vec{v}\in \mathcal{K}(\mathcal{C}) \cup \mathcal{D} \}$,
	where $pP(\vec{u}, \vec{v})$ is defined as \eqref{eq: Phi(u+v)=Phi(u)+Phi(v)+Phi(pP(u,v))} in Lemma \ref{phi(u+v) = phi(u)+phi(v)+...}.
	Obviously, $\mathcal{K}(\mathcal{C}) \cup \mathcal{D}$ is cyclic, not necessarily linear, then $\mathcal{D}'$ is cyclic.
	Moreover, for any $\vec{u}, \vec{v}\in \mathcal{D}'$, $pP(\vec{u}, \vec{v}) \in \mathcal{D}'$, then $\Phi(\mathcal{D}')$ is linear,
	leading to a contradiction since we assume $\mathcal{D}$ is maximal.
	
	It's easy to check that $\Phi(\vec{v}) + \Phi(\mathcal{C}) = \Phi(\mathcal{C})$ for any $\vec{v} \in \cap_i \mathcal{C}_i$,
	then $\cap_i \mathcal{C}_i \subseteq \mathcal{K}(\mathcal{C})$.
	Since $\mathcal{D}_i \subseteq \mathcal{C}_i$, and both of $\mathcal{D}_i$ and $\mathcal{C}_i$ have linear Gray images,
	then $\mathcal{K}(\mathcal{C}) \subseteq \mathcal{D}_i \subseteq \mathcal{C}_i$.
	Hence, we have \[ \bigcap_{i=1}^s \mathcal{C}_i \subseteq \mathcal{K}(\mathcal{C}) \subseteq \bigcap_{i=1}^s \mathcal{D}_i \subseteq \bigcap_{i=1}^s \mathcal{C}_i. \]
	Thus, the kernel of a $p$-ary code is the intersection of all of its maximal linear subspaces.
\end{proof}

\begin{theorem} \label{thm: generators of K(C)}
	Let $\mathcal{C} = \langle (a,0), (b, fh+pf) \rangle$ be a $\mathbb{Z}_p\mathbb{Z}_{p^2}$-additive cyclic code, where $fhg = x^\beta-1$ and $\gcd(\beta,p) = 1$.
	Assume that $k_1, \cdots, k_s$ are all the divisors of $g$ with minimum degree such that the Gray image of $\langle (a, 0), (b_i, fhk_i + pf)\rangle$ is linear,
	where $b_i \equiv \overline{k_i}b + (1-\overline{k_i}) \overline{\mu} b \overline{g} \pmod{a}$ and $\lambda h +\mu g = 1$.
	Then \[ \mathcal{K}(\mathcal{C}) = \langle (a, 0), (b', fhk' + pf)\rangle, \]
	where $k' = \lcm(k_1, \cdots, k_s)$ and $b' \equiv \overline{k}' b + (1-\overline{k}' \overline{\mu} b \overline{g}) \pmod{a}$.
\end{theorem}
\begin{proof}
	Let $\mathcal{D}_i = \langle (a, 0), (b_i, fhk_i + pf)\rangle$ be a cyclic subcode of $\mathcal{C}$ for $i=1, \cdots, s$.
	Since $k_i$ is the polynomial of minimum degree dividing $g$, then $\mathcal{D}_i$ is a maximal cyclic subcode of $\mathcal{C}$ with the linear Gray image.
	Moreover, we can extend $\mathcal{D}_i$ to $\mathcal{C}_i$, the maximal subcode of $\mathcal{C}$, not necessarily cyclic, with linear Gray image.
	It's possible that $\mathcal{D}_i = \mathcal{C}_i$ for some $i$.
	By Lemma \ref{lemma: K(C) is the intersection of maximal subcodes}, we know $\mathcal{K}(\mathcal{C}) = \cap_i \mathcal{D}_i$.
	According to Remark \ref{remark: generators of K(C)} and Proposition \ref{prop: intersection of C_1 and C_2}, we can obtain the result.
\end{proof}

\begin{remark}
	Note that Theorem \ref{thm: generators of K(C)} can be applied to $\mathbb{Z}_{p^2}$-additive cyclic codes, when $\alpha = 0$.
	Moreover, one can give a better result if the necessary and sufficient condition for the Gray image to be linear is found.
\end{remark}

\section{Rank of $\mathbb{Z}_p\mathbb{Z}_{p^2}$-Additive Cyclic Codes}\label{sec:rank}

Denote the linear span of $C=\Phi(\mathcal{C})$ by $\langle C \rangle$. The dimension of $\langle C \rangle$ is called the \textit{rank} of the code $C$, denoted by \textit{rank}$(C)$.

\begin{theorem} \label{thm: generators of <Phi(C)>}
	Let $\mathcal{C}$ be a $\mathbb{Z}_p\mathbb{Z}_{p^2}$-additive code of type $(\alpha, \beta; \gamma, \delta; \kappa)$ and $C=\Phi(\mathcal{C})$.
	Let $\mathcal{G}$ be the generator matrix of $\mathcal{C}$ as in \eqref{eq:generator matrix} and let $\{ \vec{u}_i \}_{i=1}^{\gamma}$ be the rows of order $p$
	and $\{ \vec{v}_j \}_{j=1}^{\delta}$ be the rows of order $p^2$ in $\mathcal{G}$.
	Then $\langle C \rangle$ is generated by $\{ \Phi(\vec{u}_i) \}_{i=1}^{\gamma}$, $\{ \Phi(\vec{v}_j), \Phi(p \vec{v}_j) \}_{j=1}^{\delta}$
	and $\{ \Phi(pP(\vec{w}_1, \vec{w}_2)) : \vec{w}_1, \vec{w}_2 \in \mathcal{C} \}$.
\end{theorem}
\begin{proof}
	If $\vec{x}\in \mathcal{C}$, then $\vec{x}$ can be expressed as $\vec{x} = \vec{v}_{j_1} + \cdots + \vec{v}_{j_m} + \vec{w}$,
	where $\{j_1, \cdots, j_m\} \subseteq \{1, \cdots, \delta \}$ and $\vec{w}$ is a codeword of order $p$ or $\vec{w} = \vec{0}$.
	Note that $j_s$ and $j_t$ may be equal.
	By Lemma \ref{phi(u+v) = phi(u)+phi(v)+...}, we have $\Phi(\vec{x}) = \Phi(\vec{v}_{j_1} + \cdots + \vec{v}_{j_m}) + \Phi(\vec{w})$,
	where $\Phi(\vec{w})$ is a linear combination of $\{ \Phi(\vec{u}_i) \}_{i=1}^{\gamma}$, $\{\Phi(p\vec{v}_j)\}_{j=1}^{\delta}$, $\{\Phi(2p\vec{v}_j)\}_{j=1}^{\delta}$,
	$\cdots$, $\{\Phi((p-1) p\vec{v}_j)\}_{j=1}^{\delta}$.
	It's easy to check that $\Phi(i p\vec{v}_j) = i \Phi(p\vec{v}_j)$, where $i$ is a positive integer.
	Moreover,
	\[
	\Phi(\vec{v}_{j_1} + \cdots + \vec{v}_{j_m}) = \Phi(\vec{w}_1 + \vec{w}_2) = \Phi(\vec{w}_1) + \Phi(\vec{w}_2) + \Phi(pP(\vec{w}_1, \vec{w}_2)),
	\]
	where $\vec{w}_1$ and $\vec{w}_2$ are just two codewords in $\mathcal{C}$ that satisfy $\vec{w}_1 + \vec{w}_2 = \vec{v}_{j_1} + \cdots + \vec{v}_{j_m}$.
	Thus, $\Phi(\vec{x})$ is generated by $\{ \Phi(\vec{u}_i) \}_{i=1}^{\gamma}$, $\{ \Phi(\vec{v}_j), \Phi(p \vec{v}_j) \}_{j=1}^{\delta}$
	and $\{ \Phi(pP(\vec{w}_1, \vec{w}_2)) : \vec{w}_1, \vec{w}_2 \in \mathcal{C} \}$.
\end{proof}

\begin{proposition} \label{prop: generators of R(C)}
	Let $\mathcal{C}$ be a $\mathbb{Z}_p\mathbb{Z}_{p^2}$-additive code of type $(\alpha, \beta; \gamma, \delta; \kappa)$ and $C=\Phi(\mathcal{C})$.
	Let $\mathcal{G}$ be the generator matrix of $\mathcal{C}$ as in \eqref{eq:generator matrix} and let $\{ \vec{u}_i \}_{i=1}^{\gamma}$ be the rows of order $p$
	and $\{ \vec{v}_j \}_{j=1}^{\delta}$ be the rows of order $p^2$ in $\mathcal{G}$.
	Then, $\mathcal{R}(\mathcal{C}) = \langle \{\vec{u}_i\}_{i=1}^{\gamma}, \{\vec{v}_j\}_{j=1}^{\delta},
	\{ pP(\vec{w}_1, \vec{w}_2): \vec{w}_1, \vec{w}_2 \in \mathcal{C} \} \rangle$.
\end{proposition}
\begin{proof}
	Denote the code $\langle \{\vec{u}_i\}_{i=1}^{\gamma}, \{\vec{v}_j\}_{j=1}^{\delta}, \{ pP(\vec{w}_1, \vec{w}_2): \vec{w}_1, \vec{w}_2 \in \mathcal{C} \} \rangle$
	by $\mathcal{D}$, then we prove that $\mathcal{D}$ is the minimum $\mathbb{Z}_p \mathbb{Z}_{p^2}$-additive code containing $\mathcal{C}$ with the linear Gray image.
	
	Obviously, for any $\vec{x}, \vec{y}\in \mathcal{D}$, we have $pP(\vec{x}, \vec{y})\in \mathcal{D}$.
	Thus, the Gray image of $\mathcal{D}$ is linear by Lemma \ref{lemma: Phi(C) is linear iff pP(u,v) in C}.
	Let $\mathcal{D}'$ be another $\mathbb{Z}_p \mathbb{Z}_{p^2}$-additive code containing $\mathcal{C}$ with linear Gray image.
	If $\vec{w}_1 \in \mathcal{C} \subseteq \mathcal{D}'$, then, by Lemma \ref{lemma: Phi(C) is linear iff pP(u,v) in C},
	we know that $pP(\vec{w}_1, \vec{w}_2) \in \mathcal{D}'$ for all $\vec{w}_2 \in \mathcal{C}$.
	Hence, $\mathcal{D} \subseteq \mathcal{D}'$ and $\mathcal{D}$ is the minimum code.
	
	By definition, $\mathcal{R}(\mathcal{C}) = \{ \vec{v}\in \mathbb{Z}_p^{\alpha} \times \mathbb{Z}_{p^2}^{\beta} \mid \Phi(\vec{v}) \in \langle \Phi(\mathcal{C}) \rangle \}$
	is the minimum $\mathbb{Z}_p \mathbb{Z}_{p^2}$-additive code containing $\mathcal{C}$ whose Gray image is linear.
\end{proof}

\begin{lemma} \label{R(C)_Y = R(C_Y)}
	Let $\mathcal{C}$ be a $\mathbb{Z}_p \mathbb{Z}_{p^2}$-additive code, then $\mathcal{R}(\mathcal{C})_Y = \mathcal{R}(\mathcal{C}_Y)$.
\end{lemma}
\begin{proof}
	Let $\mathcal{G}$ be the generator matrix of $\mathcal{C}$ as in \eqref{eq:generator matrix 2} and let $\{ \vec{u}_i = (u_i , u'_i) \}_{i=1}^{\gamma}$ be the rows of order $p$
	and $\{ \vec{v}_j = (v_j , v'_j )\}_{j=1}^{\delta}$ be the rows of order $p^2$ in $\mathcal{G}$.
	By Proposition \ref{prop: generators of R(C)}, $\mathcal{R}(\mathcal{C}) = \langle \mathcal{C}, \{p P(\vec{w}_1, \vec{w}_2): \vec{w}_1, \vec{w}_2 \in \mathcal{C} \} \rangle$.
	Since $p P(\vec{w}_1, \vec{w}_2) = (\vec{0}, p P(w''_1, w''_2))$, where $\vec{w}_1 = (w'_1,w''_1)$ and $\vec{w}_2 = (w'_2, w''_2)$,
	then
	\[
	\begin{aligned}
		\mathcal{R}(\mathcal{C})_Y & = \langle \mathcal{C}, \{p P(\vec{w}_1, \vec{w}_2): \vec{w}_1, \vec{w}_2 \in \mathcal{C} \} \rangle_Y \\
		& = \langle \mathcal{C}_Y, \{p P(w''_1, w''_2): w''_1, w''_2 \in \mathcal{C}_Y \} \rangle = \mathcal{R}(\mathcal{C}_Y),
	\end{aligned}
	\]
	and the proof is completed.
\end{proof}

\begin{theorem} \label{thm: R(C) is cyclic}
	Let $\mathcal{C}$ be a $\mathbb{Z}_p \mathbb{Z}_{p^2}$-additive cyclic code, then $\mathcal{R}(\mathcal{C})$ is also a $\mathbb{Z}_p \mathbb{Z}_{p^2}$-additive cyclic code.
\end{theorem}
\begin{proof}
	By Proposition \ref{prop: generators of R(C)}, for any $\vec{x} \in \mathcal{R}(\mathcal{C})$, then $\vec{x} = \vec{u} + \lambda p P(\vec{w}_1, \vec{w}_2)$
	for some $\vec{u}, \vec{w}_1, \vec{w}_2 \in \mathcal{C}$, where $\lambda \in \mathbb{Z}_{p}$.
	Since $\mathcal{C}$ is a $\mathbb{Z}_p \mathbb{Z}_{p^2}$-additive cyclic code, then $\pi(\vec{u}) \in \mathcal{C}$.
	It's easy to check that $\pi(\vec{a} + \vec{b}) = \pi(\vec{a}) + \pi(\vec{b})$, then
	\[ \pi\left(\vec{u} + \lambda p P(\vec{w}_1, \vec{w}_2) \right) = \pi(\vec{u}) + \lambda \pi\left( p P(\vec{w}_1, \vec{w}_2) \right). \]
	It's sufficient to check that
	\[ \pi\left( p P(\vec{w}_1, \vec{w}_2) \right) = p P(\pi(\vec{w}_1), \pi(\vec{w}_2)) \in \mathcal{R}(\mathcal{C}). \]
	By Lemma \ref{phi(u+v) = phi(u)+phi(v)+...}, we know $P(x_1,x_2) \in \mathbb{Z}_p[x_1, x_2]$, then the equality holds obviously.
	Since $\vec{w}_1, \vec{w}_2 \in \mathcal{C}$, then $\pi(p P(\vec{w}_1, \vec{w}_2)) \in \mathcal{R}(\mathcal{C})$ by Lemma \ref{phi(u+v) = phi(u)+phi(v)+...}.
	Thus, $\pi(\vec{x}) \in \mathcal{R}(\mathcal{C})$ means $\mathcal{R}(\mathcal{C})$ is cyclic.
\end{proof}

\begin{theorem} \label{thm: rank of the code}
	Let $\mathcal{C}$ be a $\mathbb{Z}_p \mathbb{Z}_{p^2}$-additive cyclic code with the generator matrix $\mathcal{G}$ in the form \eqref{eq:generator matrix 2}
	and let $\mathcal{C}'$ be the subcode generated by the matrix $\mathcal{G}'$ in the form \eqref{eq:generator matrix of C'}, then
	\[ rank(\Phi(\mathcal{C})) = \kappa_1 + \kappa_2 + rank(\phi(\mathcal{C}'_Y)). \]
\end{theorem}
\begin{proof}
	Let $\{ \vec{u}_i = (u_i, u'_i) \}_{i=1}^{\gamma}$ be the first $\gamma$ rows of $\mathcal{G}$
	and $\{ \vec{v}_j = (v_j, v'_j) \}_{j=1}^{\delta}$ be the last $\delta$ rows of $\mathcal{G}$, respectively.
	Then, $\mathcal{C}' = \langle \{ \vec{u}_i \}_{i=\kappa + 1}^{\gamma}, \{ \vec{v}_j \}_{j=1}^{\delta} \rangle$ since $\kappa = \kappa_1 + \kappa_2$.
	Define the code $\mathcal{D} = \langle \{ \vec{u}_i \}_{i=1}^{\kappa} \rangle$.
	Obviously, by Proposition \ref{prop: generators of R(C)}, $\mathcal{R}(\mathcal{D}) = \langle \{ \vec{u}_i \}_{i=1}^{\kappa} \rangle$,
	\[
	\mathcal{R}(\mathcal{C}') = \langle \{ \vec{u}_i \}_{i=\kappa + 1}^{\gamma}, \{ \vec{v}_j \}_{j=1}^{\delta},
	\{ pP(\vec{w}_1,\vec{w}_2): \vec{w}_1, \vec{w}_2 \in \mathcal{C}' \} \rangle.
	\]
	Note that for any codeword $\vec{w}_i \in \mathcal{C}$, $\vec{w}_i$ can be written as $\vec{w}_i = \vec{w}_i' + \vec{w}_i''$,
	for some $\vec{w}_i' \in \mathcal{D}$ and $\vec{w}_i'' \in \mathcal{C}'$, $i=1,2$.
	By Remark \ref{remark: values of P(a,b)}, \[ pP(\vec{w}_1, \vec{w}_2) = pP(\vec{w}''_1, \vec{w}''_2), \]
	then \[ \mathcal{R}(\mathcal{C}) = \langle \mathcal{R}(\mathcal{C}'), \mathcal{R}(\mathcal{D}) \rangle. \]
	Considering that $\mathcal{R}(\mathcal{D}) \cap \mathcal{R}(\mathcal{C}') = \{ \vec{0} \}$ since $pP(\vec{w}_1, \vec{w}_2) = (\vec{0}, pP(w'_1, w'_2))$,
	then, \[ \text{rank}(\Phi(\mathcal{C})) = \text{rank}(\Phi(\mathcal{D})) + \text{rank}(\Phi(\mathcal{C}')) = \kappa + \text{rank}(\Phi(\mathcal{C}')). \]
	Besides, by Lemma \ref{R(C)_Y = R(C_Y)}, $\text{rank}(\Phi(\mathcal{C}')) = \text{rank}(\phi(\mathcal{C}'_Y))$ and the equality holds.
\end{proof}

In Section \ref{sec:kernel}, we discussed the kernel of a $\mathbb{Z}_p \mathbb{Z}_{p^2}$-additive cyclic code.
It's shown that there doesn't always exist a $\mathbb{Z}_p \mathbb{Z}_{p^2}$-additive cyclic code for a given type.
Similarly, there may not exist a $\mathbb{Z}_p \mathbb{Z}_{p^2}$-additive cyclic code for all possible values of the rank for a given type.

\begin{example} \label{exam: z3z9-rank-1}
	\cite[Example 17]{wsk_z3z9_linear_rank_kernel}
	By Remark \ref{remark: conditions of type}, there exists a $\mathbb{Z}_3 \mathbb{Z}_{9}$-additive code $\mathcal{C}$ of type $(\alpha, 13; 2, 4; 1)$.
	The possible dimensions of the kernel are $10$, $11$, $12$, $13$, $14$, $15$, $16$, $17$, $18$.
	But just as Example \ref{exam: z3z9-kernel-1} shows, the $\mathbb{Z}_3 \mathbb{Z}_{9}$-additive cyclic code of type $(\alpha, 13; 2, 4; 1)$ does not exist.
\end{example}

\begin{proposition} \label{prop: R(C)=fh + pf/r over Zp2}
	Let $\mathcal{C} = \langle fh + pf \rangle$ be a cyclic code over $\mathbb{Z}_{p^2}$ of length $n$ with $\gcd(\beta, p) = 1$ and $fhg = x^\beta - 1$.
	Then $\mathcal{R}(\mathcal{C})$ is cyclic and $\mathcal{R}(\mathcal{C}) = \langle fh + p\frac{f}{r} \rangle$ for some $r \mid f$.
\end{proposition}
\begin{proof}
	If $\mathcal{D} = \langle (a, 0), (b, fh+pf)\rangle$ is a $\mathbb{Z}_p\mathbb{Z}_{p^2}$-additive cyclic code, then $\mathcal{D}_Y = \mathcal{C}$.
	By Theorem \ref{thm: R(C) is cyclic}, we can get $\mathcal{R}(\mathcal{C})$ is a linear cyclic code over $\mathbb{Z}_{p^2}$.
	Without loss of generality, let $\mathcal{R}(\mathcal{C}) = \langle f'h' + pf' \rangle$ with $f'h'g' = x^n - 1$.
	According to Lemma \ref{R(C)_Y = R(C_Y)}, we have that $\mathcal{R}(\mathcal{D}_Y) = \mathcal{R}(\mathcal{C}) = (\mathcal{R}(\mathcal{D}))_Y$.
	By Proposition \ref{prop: generators of R(C)}, it's clear that $\mathcal{C}$ and $\mathcal{R}(\mathcal{C})$ have the same generators of order $p^2$.
	Thus, $fh = f'h'$, which implies that $g = g'$.
	Then by Lemma \ref{lemma: f'|f on Zp2}, we have $f' \mid f$ and $f = f' r$ for some $r$.
	Hence, \[ \mathcal{R}(\mathcal{C}) = \langle f'h', pf' \rangle = \langle fh, pf_r \rangle, \] where $f_r = f/r$.
\end{proof}

\begin{remark}
	Keep the notations in Proposition \ref{prop: R(C)=fh + pf/r over Zp2}.
	If $r$ is the Hensel lift of $\gcd\left( \overline{f}, \overline{g}_m \right)$, then $\overline{f'} = \overline{f}/\overline{r}$ and $\overline{g}_m$ are coprime,
	i.e., $\gcd\left( \overline{f}/\overline{r}, \overline{g}_m \right) = 1$.
	Thus, by Theorem \ref{thm: linearity of images of ZpZp2 cyclic codes}, if $g_2 \mid g_3 \mid \cdots \mid g_m$, then $\langle fh, pf/r \rangle$ has linear Gray image.
	It's known that $\mathcal{R}(\mathcal{C})$ is the minimum cyclic code with linear Gray image over $\mathbb{Z}_{p^2}$ containing $\mathcal{C}$,
	then $\mathcal{R}(\mathcal{C}) \subseteq \langle fh + p\frac{f}{r} \rangle$.
	Considering that $\mathcal{R}(\mathcal{C}) = \langle fh, pf_r \rangle$ is cyclic, if $r$ is the polynomial of minimum degree dividing $f$,
	then $\mathcal{R}(\mathcal{C}) = \langle fh + p\frac{f}{r} \rangle$.
\end{remark}

\begin{lemma} \label{lemma: properties of generators of R(C)}
	Let $\mathcal{C} = \langle (a, 0), (b, fh+pf) \rangle$ be a $\mathbb{Z}_p \mathbb{Z}_{p^2}$-additive cyclic code, where $fhg = x^\beta-1$ and $\gcd(\beta, p) = 1$.
	Then, $\mathcal{R}(\mathcal{C}) = \langle (a_r, 0), (b_r, fh+p\frac{f}{r}) \rangle$, where $r$ divides $f$ and $a_r$ divides $a$.
\end{lemma}
\begin{proof}
	By Theorem \ref{thm: R(C) is cyclic}, let $\mathcal{R}(\mathcal{C}) = \langle (a_r, 0), (b_r, f_r h_r + pf_r)\rangle$ be a $\mathbb{Z}_p\mathbb{Z}_{p^2}$-additive
	cyclic code.
	Since $(a, 0)\in \mathcal{R}(\mathcal{C})$, then $a_r$ divides $a$.
	By Proposition \ref{prop: generators of R(C)}, we know that $\mathcal{C}$ and $\mathcal{R}(\mathcal{C})$ have the same generators of order $p^2$.
	Since $\mathcal{C} \subseteq \mathcal{R}(\mathcal{C})$ and $\mathcal{C}_Y \subseteq (\mathcal{R}(\mathcal{C}))_Y$,
	then, according to Proposition \ref{prop: R(C)=fh + pf/r over Zp2}, $g_r = g$ and thus, $f_rh_r = fh$.
	Moreover, by Proposition \ref{prop: f|f_k}, we have $f_r$ divides $f$.
	Thus, there exists a polynomial $r\in \mathbb{Z}_{p^2}[x]$ such that $f = r f_r$ and $h_r = hr$.
\end{proof}

\begin{theorem} \label{thm: generators of R(C)}
	Let $\mathcal{C} = \langle (a, 0), (b, fh+pf) \rangle$ be a $\mathbb{Z}_p \mathbb{Z}_{p^2}$-additive cyclic code, where $fhg = x^\beta-1$ and $\gcd(\beta, p) = 1$.
	If $r$ is the polynomial of minimum degree dividing $f$ such that $\langle fh + pf/r \rangle$ has linear Gray image,
	then \[ \mathcal{R}(\mathcal{C}) = \langle (a', 0), (b', fh + pf/r)\rangle, \]
	where $\overline{g}_m$ is the $m_{th}$ circle product of $\overline{g}$, $a' = \gcd(a, \overline{\mu}\overline{g} b)$, $b' = b - \overline{\mu}\overline{g} b$
	and $\lambda h + \mu g = 1$.
\end{theorem}
\begin{proof}
	Obviously, the Gray image of $\mathcal{R}(\mathcal{C})_Y = \mathcal{R}(\mathcal{C}_Y) = \langle fh, pf/r \rangle$ is linear.
	Let $\mathcal{D} = \langle \mathcal{C}, (0, pf/r)$, then
	\[
	\begin{aligned}
		\mathcal{D} & = \langle (a, 0), (b', fh), (\overline{\mu} b \overline{g}, pf), (0, pf/r) \rangle \\
		& = \langle (a, 0), (\overline{\mu} b \overline{g}, 0), (b', fh + pf/r) \rangle \\
		& = \langle (a', 0), (b', fh + pf/r) \rangle,
	\end{aligned}
	\]
	where $a' = \gcd(a, \overline{\mu}\overline{g} b)$, $b' = b - \overline{\mu}\overline{g} b$ and $\lambda h + \mu g = 1$.
	
	It's known that $pP(\vec{w}_1, \vec{w}_2) = (\vec{0}, pP(w'_1, w'_2))$ for any $\vec{w}_1, \vec{w}_2 \in \mathcal{C}$.
	By the definition of $\mathcal{R}(\mathcal{C})$, $pP(w'_1, w'_2) \in \mathcal{R}(\mathcal{C})_Y = \mathcal{R}(\mathcal{C}_Y) = \langle fh, pf/r \rangle = \mathcal{D}_Y$.
	Note that $(\mathcal{D}_Y)_p = \langle pfh + pf/r \rangle = \langle pf/r \rangle$, then $(\vec{0}, pP(v'_j, v'_k)) \in \langle (0, pf/r) \rangle$.
	Thus, $pP(\vec{w}_1, \vec{w}_2) \in \mathcal{D}$ since $\mathcal{D} = \langle \mathcal{C}, (0, pf/r)$,
	which implies that $\mathcal{R}(\mathcal{C}) \subseteq \mathcal{D}$.
	Besides, $r$ is the polynomial with minimum degree dividing $f$ satisfying that $\langle fh+pf/r \rangle$ has linear Gray image, then $\mathcal{R}(\mathcal{C}) = \mathcal{D}$.
\end{proof}

\section{Kernel and Rank of Some Special Cyclic Codes}\label{sec:pairs and classification}

Let $f$ be a polynomial over $\mathbb{Z}_{p^2}$, then $f$ is called \textit{basic irreducible} if $\overline{f} \equiv f \pmod{p}$ is irreducible over $\mathbb{Z}_p$.
Let us recall Hensel's Lemma and Hensel's lift.
\begin{lemma} \cite[Theorem 13.11]{Hensel_lemma} \label{Hensel's Lemma}
	(Hensel's Lemma)
	Let $f$ be a monic polynomial in $\mathbb{Z}_{p^2}[x]$ and assume that $\overline{f} = g_1 g_2 $ over $\mathbb{Z}_p$, where $g_1$ and $g_2$ are coprime monic polynomials.
	Then there exist coprime monic polynomials $f_1$ and $f_2$ over $\mathbb{Z}_{p^2}$ such that $f=f_1f_2$ and $\overline{f_1} = g_1$, $\overline{f_2} = g_2$.
	$f \in \mathbb{Z}_{p^2}[x]$ is called the Hensel lift of $g \in \mathbb{Z}_p[x]$ if $\overline{f} = g$ and $f$ divides $x^n-1$ for some $n$ with $\gcd(n,p)=1$.
\end{lemma}

For a $\mathbb{Z}_p\mathbb{Z}_{p^2}$-additive code $\mathcal{C}$,
let $ker(\mathcal{C}) = \dim(\ker(\Phi(\mathcal{C})))$ and $rank(\mathcal{C}) = \dim(\langle \Phi(\mathcal{C}) \rangle)$.
In this section, we will show that there does not exist a $\mathbb{Z}_p\mathbb{Z}_{p^2}$-additive cyclic code
for all possible values of $rank(\mathcal{C})$ and $ker(\mathcal{C})$.

\subsection{Some Special Factorizations of $x^n - 1$}\label{subsec:factors of x^n-1}
Recall some special factorizations of $x^n-1$ over $\mathbb{F}_p$ mentioned in \cite{kernel_of_4-ary_code} and \cite{Xuan-Gray-images}.

Let $n$ and $m$ be two positive integers such that $\gcd(n,m) = 1$,
then the least positive integer $k$ for which $m^{k} \equiv 1 \pmod{n}$ is called the \textit{multiplicative order} of $m$ modulo $n$.

\begin{enumerate}
	\renewcommand{\labelenumi}{(\theenumi)}
	\item If $p$ is a primitive root modulo $n$ and $n$ is prime, then \[ x^n - 1 = (x-1) (x^{n-1} + x^{n-2} + \cdots + x + 1), \]
	where the two factors are both irreducible over $\mathbb{F}_p$.
	
	\item Assume $n = q^2$ with the prime $q$ and $x^n - 1 = (x-1)ab$,
	where $a = Q_q(x) = 1 + x + \cdots + x^{q-1}$ and $b = Q_{q^2}(x) = 1 + x^q + x^{2q} + \cdots + x^{(q-1)q}$ are both basic irreducible polynomials over $\mathbb{Z}_{p^2}$.
	
	\item Let $n$ be prime and $\ell = (n-1)/2$ be the multiplicative order of $p$ modulo $n$, then over $\mathbb{Z}_{p^2}$, we have $x^n - 1 = (x-1) f_1 f_2$,
	where $\overline{f_1}$ and $\overline{f_2}$ are two irreducible polynomials of degree $\ell$ over $\mathbb{F}_p$.
	Moreover, $\overline{f}_i \otimes \overline{f}_i = \frac{x^n-1}{x-1}$,
	$\overline{f}_i \otimes \overline{f}_i \otimes \overline{f}_i = x^n -1$ and $\gcd(\overline{f_i}, \overline{f_j} \otimes \overline{f_j}) = \overline{f_i}$, where $i=1,2$.
\end{enumerate}
\subsection{Rank and Kernel of Some Special Cyclic Codes over $\mathbb{Z}_{p^2}$}\label{subsec:class 1}
For a cyclic code $\mathcal{C} = \langle (fh + pf) \rangle$ over $\mathbb{Z}_{p^2}$, it can be regarded as a $\mathbb{Z}_p \mathbb{Z}_{p^2}$-additive cyclic code with $\alpha = 0$.
Then as for its type, $\delta = \deg(g)$ and $\gamma = \deg(h)$.

\begin{lemma}\label{lemma: rank=gamma + 2delta + deg(r) over zp2}
	Let $\mathcal{C} = \langle (fh + pf) \rangle$ be a cyclic code over $\mathbb{Z}_{p^2}$ of length $n$, where $\gcd(n, p) = 1$ and $fhg = x^n- 1$.
	Let $r$ be the polynomial such that $\mathcal{R}(\mathcal{C}) = \langle fh + pf/r \rangle$,
	then the dimension of $\mathcal{R}(\mathcal{C})$ is $rank(\mathcal{C}) = \gamma + 2\delta + \deg(r)$.
	Moreover, the dimension of the kernel is $ker(\mathcal{C}) = \gamma + 2 \delta - \deg(k)$, where $k$ divides $g$.
\end{lemma}
\begin{proof}
	By Proposition \ref{prop: R(C)=fh + pf/r over Zp2}, we have $\mathcal{R}(\mathcal{C}) = \langle (f'h' + pf') \rangle$ with $f' = f/r$, $h' = hr$ and $g' = g$.
	Then the size of $\mathcal{R}(\mathcal{C})$ is $(p^2)^{\deg(g')} p^{\deg(h')} = p^{2\deg(g) + \deg(h) + \deg(r)}$,
	i.e., $rank(\mathcal{C}) = \gamma + 2\delta + \deg(r)$.
	As for $ker(\mathcal{C})$, it follows the proof of \cite[Corollary 4]{kernel_of_4-ary_code}.
\end{proof}

\begin{proposition}
	Let $\mathcal{C} = \langle (fh + pf) \rangle$ be a cyclic code over $\mathbb{Z}_{p^2}$ of length $n$, where $\gcd(n, p) = 1$ and $fhg = x^n- 1$.
	Assume that $n$ is prime and $p$ is a primitive root modulo $n$.
	If $g = Q_n(x) = x^{n-1} + \cdots + 1$ and $h = 1$, then $ker(\mathcal{C}) = \gamma + \delta = n$ and $rank(\mathcal{C}) = \gamma + 2\delta + 1$.
	In all other cases, the Gray image of $\mathcal{C}$ is linear, i.e., $rank(\mathcal{C}) = ker(\mathcal{C}) = \gamma + 2\delta$.
\end{proposition}
\begin{proof}
	If $g = Q_n(x) = x^{n-1} + \cdots + 1$ and $h = 1$, then $\mathcal{C} = \langle x - 1 \rangle$ and the generator matrix of $\mathcal{C}$ in standard form is
	\[
	G =
	\begin{pmatrix}
		1 & 0 & 0 & \cdots & 0 & -1 \\
		0 & 1 & 0 & \cdots & 0 & -1 \\
		0 & 0 & 1 & \cdots & 0 & -1 \\
		\vdots & \vdots & \vdots &  & \vdots & \vdots \\
		0 & 0 & 0 & \cdots & 1 & -1 \\
	\end{pmatrix}.
	\]
	Without loss of generality, let $\vec{v}_i$ be the rows of $G$, $1 \leqslant i \leqslant n-1$.
	Note that \[ pP(\vec{v}_i, \vec{v}_j) = (0, 0, \cdots, 0, p) \notin \mathcal{C}, \]
	since $P(-1, -1) = P(p^2-1, p^2-1) = P(p-1, p-1) = 1$.
	Hence, the Gray image of $\mathcal{C}$ is not linear.
	Since $g$ is irreducible, i.e., the divisors of $g$ are just $1$ and $g$, we get $ker(\mathcal{C}) = \gamma + 2 \delta - \deg(g) = n$
	and $rank(\mathcal{C}) = \gamma + 2\delta + 1$.
	As for other cases, see \cite[Theorem 12 \& Corollary 6]{kernel_of_4-ary_code}.
\end{proof}

Then we consider the following case: there are $3$ irreducible factors of $x^n-1$.

From now on, in Table \ref{table: kernel and rank of zp2 cyclic codes 1}, assume $n = q^2$ with the prime $q$ and $x^n - 1 = (x-1)ab$,
where $a = Q_q(x) = 1 + x + \cdots + x^{q-1}$ and $b = Q_{q^2}(x) = 1 + x^q + x^{2q} + \cdots + x^{(q-1)q}$ are both irreducible.
If $n=q$ is prime,then $x^n - 1 = (x-1)ab$ with $\deg(a) = \deg(b) = d$,
where $d = (n-1)/2$ is the multiplicative order of $p$ modulo $n$.
By \cite{Xuan-Gray-images}, it's known that the Gray image of $\langle fh + pf \rangle$ is linear if and only if $f$ is coprime with $\overline{g_3}$.

\begin{table}[h]
	\caption{Parameters of cyclic codes when $n=q$ or $q^2$ }
	\label{table: kernel and rank of zp2 cyclic codes 1}
	\centering
	\begin{tabular}{c|c|c|c|c|c}
		\hline\hline
		$n$ & $f$ & $g$ & $h$ & $rank(\mathcal{C})$ & $ker(\mathcal{C})$ \\
		\hline\hline
		Any & $1$ & $\ast$ & $\ast$ & $\gamma + 2\delta$ & $\gamma + 2\delta$   \\
		\hline
		Any & $\ast$ & $1$ & $\ast$ & $\gamma + 2\delta$ & $\gamma + 2\delta$   \\
		\hline
		Any & $\ast$ & $x-1$ & $\ast$ & $\gamma + 2\delta$ & $\gamma + 2\delta$   \\
		\hline
		$q^2$ & $a$ & $(x-1)b$ & $1$ & $\gamma + 2\delta + t_a$ & $\gamma + \delta + 1$    \\
		\hline
		$q^2$ & $a$ & $b$ & $x-1$ & $\gamma + 2\delta + t_a$ & $\gamma + \delta$    \\
		\hline
		$q^2$ & $b$ & $(x-1)a$ & $1$ & $\gamma + 2\delta$ & $\gamma + 2\delta$   \\
		\hline
		$q^2$ & $b$ & $a$ & $x-1$ & $\gamma + 2\delta + t_b$ &  $\gamma + \delta$    \\
		\hline
		$q^2$ & $x-1$ & $a$ & $b$ & $\gamma + 2\delta + 1$ &  $\gamma + \delta$  \\
		\hline
		$q^2$ & $x-1$ & $b$ & $a$ & $\gamma + 2\delta + 1$ &  $\gamma + \delta$   \\
		\hline
		$q^2$ & $x-1$ & $ab$ & $1$ & $\gamma + 2\delta + 1$ &  $\gamma + \delta$   \\
		\hline
		$q^2$ & $(x-1)a$ & $b$ & $1$ & $\gamma + 2\delta + t_a + 1$ &  $\gamma + \delta$   \\
		\hline
		$q^2$ & $(x-1)b$ & $a$ & $1$ & $\gamma + 2\delta + 1$ & $\gamma + \delta$   \\
		\hline\hline
		$q$ & $a$ & $(x-1)b$ & $1$ & $\gamma + 3\delta - 1$ & $\gamma + \delta + 1$ \\
		\hline
		$q$ & $a$ & $b$ & $x-1$ & $\gamma + 3\delta$ & $\gamma + \delta$   \\
		\hline
		$q$ & $b$ & $(x-1)a$ & $1$ & $\gamma + 3\delta - 1$ & $\gamma + \delta + 1$ \\
		\hline
		$q$ & $b$ & $a$ & $x-1$ & $\gamma + 3\delta$ &  $\gamma + \delta$  \\
		\hline
		$q$ & $x-1$ & $a$ & $b$ & $\gamma + 2\delta + 1$ ($\gamma + 2\delta$) & $\gamma + \delta$ ($\gamma + 2\delta$)   \\
		\hline
		$q$ & $x-1$ & $b$ & $a$ & $\gamma + 2\delta + 1$ ($\gamma + 2\delta$) & $\gamma + \delta$ ($\gamma + 2\delta$)   \\
		\hline
		$q$ & $x-1$ & $ab$ & $1$ & $\gamma + 2\delta + 1$ & $\gamma + \delta$   \\
		\hline
		$q$ & $(x-1)a$ & $b$ & $1$ & $\gamma + 3\delta + 1$ ($\gamma + 3\delta$) & $\gamma + \delta$  \\
		\hline
		$q$ & $(x-1)b$ & $a$ & $1$ & $\gamma + 3\delta + 1$ ($\gamma + 3\delta$) & $\gamma + \delta$   \\
		\hline\hline
	\end{tabular}
\end{table}

\begin{remark}
	In Table \ref{table: kernel and rank of zp2 cyclic codes 1},
	all possible kernels and ranks of the codes obtained from the factorization $x^n-1=(x-1)ab$
	are completely determined, where $t_a = \deg(a)$ and $t_b = \deg(b)$.
	The $\ast$ means that it can take all possible values (or polynomials).
	In fact, the results in Table \ref{table: kernel and rank of zp2 cyclic codes 1} are the generalization of those obtained in \cite[Table 1]{kernel_of_4-ary_code}.
	
	Besides, $s(t)$ means the value of the $rank$ or $ker$ is $s$ when $p\geqslant 3$ and $t$ when $p=2$.
	It's interesting that the classifications of $\mathbb{Z}_2 \mathbb{Z}_4$ and $\mathbb{Z}_p \mathbb{Z}_{p^2} (p>2)$ are different.
\end{remark}

Let $n$ be prime, if $d_N=(n-1)/N$ is the multiplicative order of $p$ modulo $n$ for some positive integer $N$, then $x^n-1=(x-1)f_1\cdots f_N$,
where $\overline{f_i}$ is an irreducible polynomial of degree $d_N$ over $\mathbb{F}_p$, $1\leqslant i\leqslant N$.
Namely, when $n$ is prime, the factorization of $x^n-1$ is completely given.
As for the classification of cyclic codes of length $n$ over $\mathbb{Z}_{p^2}$, \cite[Theorem 6.34]{Finite_Fields} may be helpful.

\subsection{Rank and Kernel of Some Special Cyclic Codes over $\mathbb{Z}_{p} \mathbb{Z}_{p^2}$}\label{subsec:class 2}

Let $\mathcal{C} = \langle (a, 0), (b, fh + pf) \rangle$ be a $\mathbb{Z}_p \mathbb{Z}_{p^2}$-additive cyclic code.
By Theorem \ref{thm: generators of R(C)}, if $a$ divides $b$, then $\mathcal{R}(\mathcal{C}) = \langle (a, 0), (b', fh+pf/r) \rangle$, for some $a'$,
where $r$ divides $f$.
Then, by Lemma \ref{lemma: rank=gamma + 2delta + deg(r) over zp2}, we have
\begin{theorem} \label{thm: rank=gamma + 2delta + deg(r) over zpzp2}
	Let $\mathcal{C} = \langle (a, 0), (b, fh + pf) \rangle$ be a $\mathbb{Z}_p \mathbb{Z}_{p^2}$-additive cyclic code of type $(\alpha, \beta; \gamma, \delta; \kappa)$,
	with $\gcd(\beta, p) = 1$, $fhg = x^\beta - 1$.
	If $\mathcal{C}$ is separable, then $rank(\mathcal{C}) = \gamma + 2\delta + \deg(r)$ and $ker(\mathcal{C}) = \gamma + 2 \delta - \deg(k)$,
	where $r$ is the polynomial such that $\mathcal{R}(\mathcal{C}_Y) = \langle fh + pf/r \rangle$ and $k$ divides $g$.
\end{theorem}

\begin{theorem} \label{thm: zpzp2 possible ranks and kernels}
	Let $\mathcal{C} = \langle (a, 0), (b, fh + pf) \rangle$ be a $\mathbb{Z}_p \mathbb{Z}_{p^2}$-additive cyclic code of type $(\alpha, \beta; \gamma, \delta; \kappa)$,
	with $\gcd(\beta, p) = 1$, $fhg = x^\beta - 1$ and $\alpha = \gamma \leqslant \delta - 1$, $\beta = 2\delta + 1 \geqslant 5$.
	Assume that $\beta$ is prime and $\delta = (\beta-1)/2$ is the multiplicative order of $p$ modulo $\beta$.
	Then all the possible values of $(rank(\mathcal{C}),ker(\mathcal{C}))$ are $(\gamma+3\delta, \gamma+\delta)$, $(\gamma+3\delta+1, \gamma+\delta)$
	and $(\gamma+3\delta-1, \gamma+\delta)$.
\end{theorem}
\begin{proof}
	From Section \ref{subsec:factors of x^n-1}, it's known that there are $3$ basic irreducible factors (the Hensel's lift) of $x^\beta-1$ over $\mathbb{Z}_{p^2}$,
	namely, $x^\beta - 1 = (x-1)f_1 f_2$.
	Since either $g = f_1$ or $g = f_2$, then $\overline{g_2} \mid \overline{g_3} = x^{\beta} - 1$,
	where $\overline{g_2}$ and $\overline{g_3}$ are the second and third circle product of $\overline{g}$, respectively.
	Without loss of generality, let $g = f_1$.
	Note that $\deg(a) = \deg(h) \leqslant \alpha \leqslant \delta - 1$ since $\alpha = \gamma \leqslant \delta - 1$.
	Then $\deg(h) \leqslant 1$, which implies that $\deg(f) \geqslant \delta$ and $f_2$ divides $f$.
	Considering that $\gcd(\overline{f_j}, \overline{f_i} \otimes \overline{f_i}) = \overline{f_j}$, $\Phi(\mathcal{C})$ is nonlinear, where $i,j=1,2$.
	
	By Theorem \ref{thm: generators of K(C)}, $\mathcal{K}(\mathcal{C}) = \langle (a, 0), (b_k, fhk + pf) \rangle$, where $k$ divides $g$.
	Since $\Phi(\mathcal{C})$ is nonlinear, then $k=g$, and $\mathcal{K}(\mathcal{C}) = \langle (a, 0), (b_k, p f) \rangle$ does not contain codewords of order $p^2$,
	i.e., $ker(\mathcal{C}) = \gamma + \delta > \kappa$.
	By Theorem \ref{thm: generators of R(C)}, $\mathcal{R}(\mathcal{C}) = \langle (a_r, 0), (b_r, fh + pf/r) \rangle$,
	where $r$ divides $f$, which implies that $rank(\mathcal{C}) = 2\deg(g) + \alpha - \deg(a_r) + \deg(hr)$.
	Considering that $\deg(a) = \deg(h) \leqslant 1$ and $a_r$ divides $a$, we have
	\[
	rank(\mathcal{C}) =
	\begin{cases}
		\gamma + 2\delta + \deg(r), & \text{if} \quad \deg(a_r) = \deg(a), \\
		\gamma + 2\delta + \deg(r) - 1, & \text{if} \quad \deg(a_r) \neq \deg(a), \\
	\end{cases}
	\]
	Note that $\gcd(\overline{f_2}, \overline{f_1} \otimes \overline{f_1}) = \overline{f_2}$, and $f = (x-1)f_2$ or $f = f_2$ and $\deg(f_2) = \delta$.
	When $p=2$, $\gcd(\overline{f}, \overline{g} \otimes \overline{g}) = \overline{f_2}$, i.e., $\deg(r) = \deg(f_2) = \delta$.
	Otherwise, $\gcd(\overline{f}, \overline{g} \otimes \overline{g}) = \overline{f}$ and $\deg(r) = \deg(f) = \delta$ or $\delta + 1$.
\end{proof}

\section{Comparison and Examples}\label{sec:comparison}

Recall two theorems about pairs of rank and kernel dimension of $\mathbb{Z}_p \mathbb{Z}_{p^2}$-additive code.

\begin{theorem} \label{thm: existence of z2z4-linear codes}
	\cite[Theorem 6]{z2z4-linear}
	Let $\alpha, \beta, \gamma, \delta, \kappa$ be positive integers satisfying \eqref{eq: conditions of type}.
	Then, there exists a nonlinear $\mathbb{Z}_2 \mathbb{Z}_4$-linear code $\mathcal{C}$ of type $(\alpha, \beta; \gamma, \delta; \kappa)$
	with $k = \gamma + 2\delta - \overline{k}$ and $r = \gamma + 2\delta + \overline{r}$ if and only if
	\[
	\begin{cases}
		\overline{r} \in \left\{ 2, \cdots, \min\left( \beta - (\gamma - \kappa) - \delta, \binom{\overline{k}}{2} \right) \right\}, & \overline{k} \equiv 1 \pmod{2}, \\
		\overline{r} \in \left\{ 1, \cdots, \min\left( \beta - (\gamma - \kappa) - \delta, \binom{\overline{k}}{2} \right) \right\}, & \overline{k} \equiv 1 \pmod{2}, \\
	\end{cases}
	\]
	where $\overline{k} \in \{ 2, \cdots, \delta \}$.
\end{theorem}

\begin{theorem} \label{thm: existence of zpzp2-linear codes}
	\cite[Theorem 17]{wsk_z3z9_linear_rank_kernel}
	Let $\alpha, \beta, \gamma, \delta, \kappa$ be positive integers satisfying \eqref{eq: conditions of type}.
	Then, there is a $\mathbb{Z}_p \mathbb{Z}_{p^2}$-additive code of type $(\alpha, \beta; \gamma, \delta; \kappa)$ only if $\overline{k} = \overline{r} = 0$ or
	\[
	1\leqslant \overline{k} \leqslant \delta \quad \text{and} \quad
	1 \leqslant \overline{r} \leqslant \min{\left(\beta - (\gamma-\kappa) - \delta, \binom{\overline{k}}{2} + \binom{\overline{k}+2}{3}\right)},
	\]
	where $r = \gamma + 2\delta + \overline{r}$ and $k = \gamma + 2\delta - \overline{k}$, respectively.
\end{theorem}

\begin{remark} \label{remark: r=k=gamma + 2delta}
	$\Phi(\mathcal{C})$ is linear if and only if $\overline{k} = \overline{r} = 0$ for $\mathbb{Z}_2 \mathbb{Z}_{4}$ and $\mathbb{Z}_p \mathbb{Z}_{p^2}$-additive codes.
\end{remark}

\begin{example} \cite[Example 3 \& 7]{z2z4_cyclic_kernel}
	Let $\mathcal{C} = \langle (a, 0), (b, fh + 2f) \rangle$ be a $\mathbb{Z}_2 \mathbb{Z}_4$-additive cyclic code of type $(2, 7; 2, 3; \kappa)$.
	Let $k$ and $r$ be the dimensions of $\ker(\Phi(\mathcal{C}))$ and $\langle \Phi(\mathcal{C}) \rangle$, respectively.
	Then, by Remark \ref{remark: conditions of type}, $k\in \{ 5,6,8 \}$ and $r \in \{ 8,9,10,11 \}$.
	The existence of $\mathbb{Z}_2 \mathbb{Z}_4$-linear codes of general additive codes and additive cyclic codes
	are listed in Table \ref{table: z2z4-linear} and Table \ref{table: z2z4-cyclic}, respectively.
\end{example}

\begin{table}[h]
	\begin{minipage}[t]{0.45\textwidth}
		\caption{$\mathbb{Z}_2 \mathbb{Z}_4$-additive codes}
		\label{table: z2z4-linear}
		\centering
		\begin{tabular}{c|c|c|c|c}
			\hline\hline
			\diagbox{$k$}{$r$} & 8 & 9 & 10 & 11  \\
			\hline\hline
			8 & $\bullet$ &  &  &  \\
			\hline
			7 &  & &  &  \\
			\hline
			6 & & $\bullet$ & & \\
			\hline
			5 & & $\bullet$ & $\bullet$ & $\bullet$ \\
			\hline\hline
		\end{tabular}
	\end{minipage}
	\begin{minipage}[t]{0.45\textwidth}
		\caption{$\mathbb{Z}_2 \mathbb{Z}_4$-additive cyclic codes}
		\label{table: z2z4-cyclic}
		\centering
		\begin{tabular}{c|c|c|c|c}
			\hline\hline
			\diagbox{$k$}{$r$} & 8 & 9 & 10 & 11  \\
			\hline\hline
			8 & &  &  &  \\
			\hline
			7 &  & &  &  \\
			\hline
			6 &  & &  &  \\
			\hline
			5 &  & & & $\bullet$ \\
			\hline\hline
		\end{tabular}
	\end{minipage}
\end{table}

\begin{example}
	Let $\mathcal{C} = \langle (a, 0), (b, fh + 3f) \rangle$ be a $\mathbb{Z}_3 \mathbb{Z}_9$-additive cyclic code of type $(\alpha, 11; \gamma, \delta; \kappa)$.
	Let $k$ and $r$ be the dimensions of $\ker(\Phi(\mathcal{C}))$ and $\langle \Phi(\mathcal{C}) \rangle$, respectively.
	In $\mathbb{Z}_9$, \[ x^{11} - 1 = (x-1)(x^5+3x^4-x^3+x^2+2x-1)(x^5+7x^4-x^3+x^2+6x-1), \]
	where the factors are basic irreducible polynomials over $\mathbb{Z}_9$.
	Obviously, if $\delta \notin \{ 0,1,5,6,10,11 \}$, then the code doesn't exist.
	So we consider the codes of type $(2, 11; 2, 5; 1)$.
	Then, by Theorem \ref{thm: existence of zpzp2-linear codes}, $ker(\mathcal{C})\in \{ 7, \cdots, 12 \}$ and $rank(\mathcal{C}) \in \{ 12,\cdots,17 \}$.
	
	By Theorem \ref{thm: zpzp2 possible ranks and kernels}, $ker(\mathcal{C}) = 2+5 = 7$ and $rank(\mathcal{C}) = 17$ or $16$.
	
	Both of the codes $\mathcal{C} = \langle (x-1, 0), (0, (x-1)(x^5+3x^4-x^3+x^2+2x-1) + 3(x^5+3x^4-x^3+x^2+2x-1)) \rangle$ and
	$\mathcal{C} = \langle (x-1, 0), (0, (x-1)(x^5+7x^4-x^3+x^2+6x-1) + 3(x^5+7x^4-x^3+x^2+6x-1)) \rangle$
	are $\mathbb{Z}_3 \mathbb{Z}_9$-additive cyclic codes of type $(2, 11; 2, 5; 1)$ with $rank(\mathcal{C}) = 17$ and $ker(\mathcal{C}) = 7$.
	
	Both of the codes $\mathcal{C} = \langle (x-1, 0), (1, (x-1)(x^5+3x^4-x^3+x^2+2x-1) + 3(x^5+3x^4-x^3+x^2+2x-1)) \rangle$ and
	$\mathcal{C} = \langle (x-1, 0), (1, (x-1)(x^5+7x^4-x^3+x^2+6x-1) + 3(x^5+7x^4-x^3+x^2+6x-1)) \rangle$
	are $\mathbb{Z}_3 \mathbb{Z}_9$-additive cyclic codes of type $(2, 11; 2, 5; 1)$ with $rank(\mathcal{C}) = 16$ and $ker(\mathcal{C}) = 7$.
	
	As for the values of $(r,k)$, the existences of general $\mathbb{Z}_3 \mathbb{Z}_9$-additive codes and $\mathbb{Z}_3 \mathbb{Z}_9$-additive cyclic codes
	of type $(2, 11; 2, 5; 1)$ are listed in Table \ref{table: z3z9-linear-3} and Table \ref{table: z3z9-cyclic-3}, respectively.
\end{example}
\begin{table}[h]
	\caption{$\mathbb{Z}_3 \mathbb{Z}_9$-additive codes}
	\label{table: z3z9-linear-3}
	\centering
	\begin{tabular}{c|c|c|c|c|c|c}
		\hline\hline
		\diagbox{$k$}{$r$} & 12 & 13 & 14 & 15 & 16 & 17 \\
		\hline\hline
		12 & $\bullet$ &  &  &  & \\
		\hline
		11 &  & $\bullet$ &   &   &   &  \\
		\hline
		10 &  & $\bullet$ & $\bullet$ & $\bullet$ & $\bullet$ & $\bullet$ \\
		\hline
		9 &  & $\bullet$ & $\bullet$ & $\bullet$ & $\bullet$ & $\bullet$ \\
		\hline
		8 &  & $\bullet$ & $\bullet$ & $\bullet$ & $\bullet$ & $\bullet$ \\
		\hline
		7 &  & $\bullet$ & $\bullet$ & $\bullet$ & $\bullet$ & $\bullet$ \\
		\hline\hline
	\end{tabular}
\end{table}

\begin{table}
	\caption{$\mathbb{Z}_3 \mathbb{Z}_9$-additive cyclic codes}
	\label{table: z3z9-cyclic-3}
	\centering
	\begin{tabular}{c|c|c|c|c|c|c}
		\hline\hline
		\diagbox{$k$}{$r$} & 12 & 13 & 14 & 15 & 16 & 17 \\
		\hline\hline
		12 &  &  &  & & \\
		\hline
		11 &  &  &  & &  \\
		\hline
		10 &  &  &  & &  \\
		\hline
		9 &  &  &  & &  \\
		\hline
		8 &  &  &  & &  \\
		\hline
		7 &  &  &  & & $\bullet$ & $\bullet$ \\
		\hline\hline
	\end{tabular}
\end{table}
\section{Conclusions}\label{sec:conclusions}
In this paper, we studied the $p$-ary images (Gray images) of cyclic codes over $\mathbb{Z}_{p^2}$ and $\mathbb{Z}_p \mathbb{Z}_{p^2}$.
The main contributions of this paper are the following:
\begin{enumerate}
	\renewcommand{\labelenumi}{(\theenumi)}
	\item We give the accurate expression of the rank and dimension of the kernel of the $\mathbb{Z}_p \mathbb{Z}_{p^2}$-additive cyclic code $\mathcal{C}$, respectively
	(See Theorem \ref{thm: dimension of the kernel} and Theorem \ref{thm: rank of the code}).
	
	\item We show that both of $\mathcal{K}(\mathcal{C})$ and $\left\langle \Phi(\mathcal{C}) \right\rangle$ are $\mathbb{Z}_p \mathbb{Z}_{p^2}$-additive cyclic codes
	(See Theorem \ref{K(C) is cyclic} and Theorem \ref{thm: R(C) is cyclic}).
	Moreover, we get the generator polynomials of $\mathcal{K}(\mathcal{C})$ and $\left\langle \Phi(\mathcal{C}) \right\rangle$, respectively
	(See Theorem \ref{thm: generators of K(C)} and Theorem \ref{thm: generators of R(C)}).
	
	\item According to the algebra structures of a $\mathbb{Z}_p\mathbb{Z}_{p^2}$-additive cyclic code, we gave the accurate expression of its rank and dimension of the kernel, respectively.
	
	\item Through some special factorizations $x^n-1$, several classes of cyclic codes over $\mathbb{Z}_{p^2}$ and $\mathbb{Z}_p \mathbb{Z}_{p^2}$
	have been classified by rank and the dimension of the kernel (See Table \ref{table: kernel and rank of zp2 cyclic codes 1}).
	
	\item Compared with the general $\mathbb{Z}_p\mathbb{Z}_{p^2}$-additive code shown in \cite{wsk_z3z9_linear_rank_kernel},
	the rank and the dimension of $\mathcal{C}$ cannot take all the possible values.
	
	\item We construct a $\mathbb{Z}_p \mathbb{Z}_{p^2}$-additive cyclic code $\mathcal{C}$ of some given type
	and the number of the pair of its rank and the dimension of the kernel is at most $3$
	(See Theorem \ref{thm: zpzp2 possible ranks and kernels}).
\end{enumerate}

Besides, the readers are suggested to find the necessary and sufficient condition for the Gray image to be linear,
which is helpful for modifying the generator polynomials mentioned in Theorem \ref{thm: generators of K(C)} and Theorem \ref{thm: generators of R(C)}.


%
%



\end{document}